\documentclass[11pt]{article}

\usepackage[T1]{fontenc}
\usepackage[utf8]{inputenc}
\usepackage{lmodern}
\usepackage{microtype}
\usepackage[a4paper,margin=28mm]{geometry}
\usepackage{amsmath,amssymb,amsthm,mathtools}
\usepackage{booktabs,array}
\usepackage[dvipsnames]{xcolor}
\usepackage{enumitem}
\usepackage{hyperref}

\hypersetup{
  colorlinks=true,
  linkcolor=NavyBlue,
  urlcolor=NavyBlue,
  citecolor=NavyBlue,
  pdftitle={SparseStack Is an Optimal Oblivious Subspace Embedding},
  pdfauthor={Diar Heidary},
  pdfkeywords={oblivious subspace embeddings, sparse embeddings, SparseStack, OSNAP, interacting Fock space, Lean 4}
}

\setlist{itemsep=0.25em,topsep=0.35em}
\allowdisplaybreaks

\newtheorem{theorem}{Theorem}[section]
\newtheorem{lemma}[theorem]{Lemma}
\newtheorem{proposition}[theorem]{Proposition}
\newtheorem{corollary}[theorem]{Corollary}
\theoremstyle{definition}

\newtheorem{example}[theorem]{Example}
\theoremstyle{remark}
\newtheorem{remark}[theorem]{Remark}

\newcommand{\R}{\mathbb R}
\newcommand{\E}{\mathbb E}
\newcommand{\Prob}{\mathbb P}
\newcommand{\T}{\mathsf T}
\newcommand{\Hcal}{\mathcal H}
\newcommand{\Kcal}{\mathcal K}
\newcommand{\Fcal}{\mathcal F}
\newcommand{\Mcal}{\mathcal M}
\newcommand{\Pcal}{\mathcal P}
\newcommand{\Scal}{\mathcal S}
\newcommand{\Gcal}{\mathcal G}
\newcommand{\Ical}{\mathcal I}
\newcommand{\Ncal}{\mathcal N}
\newcommand{\Tcal}{\mathcal T}
\newcommand{\Ucal}{\mathcal U}
\newcommand{\Ghat}{\widehat G}
\newcommand{\Cband}{C_{\mathrm{band}}}
\newcommand{\iid}{\mathrm{iid}}
\newcommand{\SStack}{\mathrm{SS}}
\newcommand{\tr}{\operatorname{tr}}
\newcommand{\gr}{\operatorname{gr}}
\newcommand{\supp}{\operatorname{supp}}
\newcommand{\1}{\mathbf 1}
\newcommand{\sfA}{\mathsf A}
\newcommand{\sfR}{\mathsf R}
\newcommand{\sfC}{\mathsf C}
\newcommand{\sfX}{\mathsf X}
\newcommand{\sfH}{\mathsf H}
\newcommand{\sfW}{\mathsf W}
\newcommand{\sfK}{\mathsf K}
\newcommand{\sfI}{\mathsf I}
\title{SparseStack Is an Optimal Oblivious Subspace Embedding}
\author{Diar Heidary\\[0.3em]
  \normalsize\href{mailto:diyar.heydari83@sharif.edu}{\texttt{diyar.heydari83@sharif.edu}}%
  \thanks{Unaffiliated work.  This manuscript is the result of independent
  research and was not conducted as part of a research program at Sharif
  University of Technology, a master's thesis, or a doctoral dissertation.}}
\date{Version~2.0 --- September 2026}

\begin{document}
\maketitle

\begingroup
\setlength{\fboxrule}{1.2pt}
\setlength{\fboxsep}{8pt}
\noindent
\fcolorbox{red!75!black}{red!4}{%
  \begin{minipage}{\dimexpr\linewidth-2\fboxsep-2\fboxrule\relax}
  \small\raggedright
  \textbf{\color{red!75!black}AI usage disclosure.}
  The mathematical arguments presented here were derived by large language
  models under the author's direction, specifically ChatGPT~5.6~Sol and
  Claude~Fable~5.  This document was produced and revised in collaboration
  with ChatGPT~5.6~Sol and Claude~Fable~5.1.  The main theorem is formally verified in Lean~4.
  A comparison
  with this version is recorded in Appendix~\ref{app:lean}, and the author
  believes that the formal statements accurately represent the claims made
  here.  The author takes full responsibility for the mathematical claims,
  the manuscript--formalization correspondence, all citations and
  attributions, and any remaining errors.  For the source, Lean artifact,
  version history, and further information, visit
  \href{https://github.com/DiarHaidary/nelson-nguyen-sparse-fock}
  {\nolinkurl{github.com/DiarHaidary/nelson-nguyen-sparse-fock}}.
  The formal artifact is registered in the Palomar registry:
  \href{https://palomar-registry.org/entry?id=PALOMAR-2026-08-23-000004&version=1}
  {\nolinkurl{palomar-registry.org/entry?id=PALOMAR-2026-08-23-000004\&version=1}}.
  \end{minipage}%
}
\endgroup
\medskip

\begin{abstract}
We prove that fully independent SparseStack achieves the oblivious subspace
embedding parameters conjectured by Nelson and Nguyen (FOCS 2013):
$m=O((d+\log(1/\delta))/\varepsilon^2)$ rows and
$s=O(\log(d/\delta)/\varepsilon)$ nonzero entries per column for distortion
$\varepsilon$ and failure probability $\delta$ on any fixed
$d$-dimensional subspace, with explicit constants.  The proof turns
random-matrix concentration into a problem in finite-dimensional linear
algebra.  A coupling first
reduces the moment estimates to a model with independent finite-valued
entries.  We represent these variables by multiplication operators, so
their matrix moments become exact matrix elements of a deterministic
operator on a finite tensor product.  The central estimate bounds the
contribution of $\ell\ge1$ occupied sites sharing the external factor
$\R^d$ by $d+\ell-1$ rather than $d\ell$, yielding additive dependence on
the dimension and the moment order.
This approach controls both spectral edges without Gaussian comparison.
The main theorem has been formally verified in Lean~4.
\end{abstract}

\clearpage
\begingroup
\small
\tableofcontents
\endgroup
\clearpage

\section{Introduction}
\label{sec:intro}

\subsection{Sparse subspace embeddings}

Let $U\in\R^{n\times d}$ have orthonormal columns, and write its rows as
$u_i^\T$ with $u_i\in\R^d$, so that
\begin{equation}
  \sum_{i=1}^n u_iu_i^\T=I_d .
  \label{eq:parseval-frame}
\end{equation}
A random matrix $\Pi\in\R^{m\times n}$ is an \emph{oblivious subspace
embedding} for $d$-dimensional subspaces with distortion $\varepsilon$ and
failure probability $\delta$ if, for every such $U$,
\begin{equation}
  \Prob\bigl\{\|U^\T\Pi^\T\Pi U-I_d\|>\varepsilon\bigr\}\le\delta .
  \label{eq:ose-def}
\end{equation}
Since $x^\T(U^\T\Pi^\T\Pi U-I_d)x=\|\Pi Ux\|^2-\|x\|^2$, the complement of
the event in \eqref{eq:ose-def} is exactly the event that
$(1-\varepsilon)\|x\|^2\le\|\Pi Ux\|^2\le(1+\varepsilon)\|x\|^2$ for all
$x\in\R^d$; a bound on the Gram error controls both edges of the spectrum of
$\Pi U$.

This paper concerns the following sparse construction.  Fix integers
$s,b\ge1$ and let $H_1,\ldots,H_s\in\R^{b\times n}$ be independent
CountSketch matrices~\cite{CharikarChenFarachColton2002}.  That is, for each block
$g\in[s]$ and each column $i\in[n]$, draw a row $h_g(i)$ uniformly from
$[b]$ and a Rademacher sign $\sigma_{gi}$, all $2sn$ variables being
mutually independent, and set
\begin{equation}
  (H_g)_{ai}=\sigma_{gi}\,\1_{\{h_g(i)=a\}}\qquad(a\in[b],\ i\in[n]).
  \label{eq:countsketch-entry}
\end{equation}
Each column of each block has exactly one nonzero entry.  The
\emph{fully independent SparseStack} matrix is
\begin{equation}
  \Pi=s^{-1/2}
  \begin{bmatrix}H_1\\ \vdots\\ H_s\end{bmatrix}\in\R^{m\times n},
  \qquad m=sb ,
  \label{eq:sparsestack}
\end{equation}
so every column of $\Pi$ has exactly $s$ nonzero entries, each equal to
$\pm s^{-1/2}$.  This is the block construction of Kane and
Nelson~\cite{KaneNelson2014} and one of the two OSNAP variants of Nelson and
Nguyen~\cite{NelsonNguyen2013}; the name SparseStack follows Cama\~no,
Epperly, Meyer, and Tropp~\cite{CamanoEpperlyMeyerTropp2025} and the problem
list~\cite{AmselEtAl2026}.  ``Fully independent'' refers to the mutual
independence of all pairs $(h_g(i),\sigma_{gi})$; implementations with
limited independence are not covered by our result.

The parameters $m=O((d+\log(1/\delta))/\varepsilon^2)$ and
$s=O(\log(d/\delta)/\varepsilon)$ are the target proposed
in~\cite{NelsonNguyen2013}; the row count cannot be improved in general by
the lower bound of~\cite{NelsonNguyen2014}, and the question for the
SparseStack distribution is recorded as Problem~5.4
of~\cite{AmselEtAl2026}.

\subsection{Main result}

Throughout, set
\begin{equation}
  c_\star=(1-e^{-1})^{-1}<1.582,\qquad
  \Cband=3+\sqrt2,\qquad
  \Lambda=90c_\star^2<226 .
  \label{eq:global-constants}
\end{equation}

\begin{theorem}[Fully independent SparseStack]
\label{thm:main}
Let $n,d\ge1$, let $U\in\R^{n\times d}$ satisfy $U^\T U=I_d$, and let
$0<\varepsilon\le1$ and $0<\delta\le1$.  Define
\begin{equation}
  q=\max\Bigl\{1,\Bigl\lceil\log_2\frac d\delta\Bigr\rceil\Bigr\},
  \qquad D_q=d+2q+1,
  \label{eq:qD}
\end{equation}
and choose
\begin{equation}
  s=\Bigl\lceil\frac{\Lambda(2q+1)}{\varepsilon}\Bigr\rceil,\qquad
  M_0=\frac{\Lambda^2D_q}{\varepsilon^2},\qquad
  b=\Bigl\lceil\frac{M_0}{s}\Bigr\rceil,\qquad
  m=sb .
  \label{eq:integer-parameters}
\end{equation}
Then the SparseStack matrix \eqref{eq:sparsestack} satisfies
\begin{equation}
  \Prob\bigl\{\|U^\T\Pi^\T\Pi U-I_d\|>\varepsilon\bigr\}\le\delta ,
  \label{eq:main-prob}
\end{equation}
that is, with probability at least $1-\delta$,
\begin{equation}
  (1-\varepsilon)\|x\|^2\le\|\Pi Ux\|^2\le(1+\varepsilon)\|x\|^2
  \qquad\text{for every }x\in\R^d .
  \label{eq:ose}
\end{equation}
The parameters satisfy
\begin{equation}
  s<\frac{1356\,q}{\varepsilon},\qquad
  m<\frac{306456\,(d+q)}{\varepsilon^2}.
  \label{eq:explicit-size}
\end{equation}
\end{theorem}

In particular $m=O((d+\log(d/\delta))/\varepsilon^2)$ and
$s=O(\log(d/\delta)/\varepsilon)$.  Since $\log d\le d$ for $d\ge1$,
the row bound is equivalently
$m=O((d+\log(1/\delta))/\varepsilon^2)$, as stated in the abstract.
The numerical constants have not been optimized; see
Section~\ref{sec:scope}.

\begin{remark}[Random input subspace]
\label{rem:independent-U}
The theorem is proved for a fixed matrix $U$.  If $U$ is random,
independent of $\Pi$, and satisfies $U^\T U=I_d$ almost surely, then
conditioning on $U$ and applying the theorem gives \eqref{eq:main-prob}
unchanged, since the bound is uniform in $U$.  A matrix $U$ that depends on
$\Pi$ is not covered.
\end{remark}

\subsection{Related work}

\paragraph{Sparse subspace embeddings.}
CountSketch, with a single nonzero per column, is an oblivious subspace
embedding with $m=O(d^2/\varepsilon^2)$ rows~\cite{ClarksonWoodruff2013}.
Nelson and Nguyen~\cite{NelsonNguyen2013} introduced OSNAP, with $s$
nonzeros per column, proved that $m=O(d\log^8(d)/\varepsilon^2)$ and
$s=O(\log^3(d)/\varepsilon)$ suffice, and conjectured that
$m=O((d+\log(1/\delta))/\varepsilon^2)$ and $s=O(\log(d/\delta)/\varepsilon)$
suffice; the row count is optimal by their lower
bound~\cite{NelsonNguyen2014}.  Cohen~\cite{Cohen2016} obtained
$m=O(d\log(d)/\varepsilon^2)$ with $s=O(\log(d)/\varepsilon)$ through trace
inequalities.  Chenakkod, Derezi\'nski, Dong, and
Rudelson~\cite{ChenakkodEtAl2024} obtained the optimal dimension
$m=O(d/\varepsilon^2)$ with $s=O(\log^4(d)/\varepsilon^6)$; Chenakkod,
Derezi\'nski, and Dong reduced the sparsity to
$O(\log^2(d)/\varepsilon+\log^3(d))$~\cite{ChenakkodDerezinskiDong2025} and
then, for $\varepsilon\ge d^{-O(1)}$, to $\tilde O(\log(d)/\varepsilon)$ with
$m=\tilde O(d/\varepsilon^2)$, where $\tilde O$ hides sub-polylogarithmic
factors in $d$~\cite{ChenakkodDerezinskiDong2026}.  These three works rely
on the matrix universality and comparison results of Bandeira,
Boedihardjo, and van Handel~\cite{BandeiraBoedihardjoVanHandel2023} and of
Brailovskaya and van Handel~\cite{BrailovskayaVanHandel2024}.
Theorem~\ref{thm:main} gives $m=O((d+\log(d/\delta))/\varepsilon^2)$ and
$s=O(\log(d/\delta)/\varepsilon)$ for the fully independent SparseStack
distribution, with explicit constants and without Gaussian comparison; it
does not cover other OSNAP distributions or limited independence.

\paragraph{The lower edge and lower bounds on sparsity.}
Tropp~\cite[Theorem~3.6]{Tropp2026} uses Gaussian comparison to prove the
lower-edge guarantee for the complex SparseStack construction obtained by
stacking independent complex CountSketch blocks whose nonzero entries are
Steinhaus phases.  Writing $s$ for Tropp's block count and $m$ for his total
row count, for lower distortion $\alpha\in(0,1)$ and failure probability
$p\in(0,1]$, his result requires
\[
  s\ge 6\alpha^{-1}\log(d/p),
  \qquad
  m\ge 16\alpha^{-2}\bigl(d\vee\log(d/p)\bigr),
\]
and gives
$\Prob\{\sigma_{\min}^{2}(\Phi U)>1-\alpha\}\ge 1-p$ for every
$U\in\mathbb C^{n\times d}$ with orthonormal columns.  Thus it establishes
the lower edge (injectivity), but it does not address the upper edge or the
real Rademacher model used in this paper.
Huang, Rudelson, and Tikhomirov~\cite{HuangRudelsonTikhomirov2026} show
that in the proportional regime $m=O(d)$ no sketch with $O(1)$ nonzeros
per column can be an oblivious subspace injection at constant distortion,
under a mild structural assumption; the logarithmic sparsity in
Theorem~\ref{thm:main} therefore cannot in general be replaced by a
constant in that regime.

\paragraph{Methods.}
The operator model of Section~\ref{sec:fock-model} is the finite instance of
a standard construction: multiplication by a finitely supported random
variable is a Jacobi matrix in an orthogonal-polynomial basis, and the
decomposition of that matrix into raising and lowering parts is the
interacting Fock space of Accardi and Bo\.zejko~\cite{AccardiBozejko1998}.
The use of such decompositions to compute spectral quantities by counting
vacuum-to-vacuum paths is familiar from the quantum-decomposition method in
spectral graph theory~\cite{HashimotoHoraObata2003,HoraObata2007}.  Norm
bounds for operator-valued sums in terms of their Gram operators, as in
Section~\ref{sec:band-bounds}, are in the spirit of the noncommutative
Khintchine inequalities~\cite{HaagerupPisier1993,Buchholz2001}.
Bandeira, Cipolloni, Schr\"oder, and van
Handel~\cite[Section~4.1]{BandeiraCipolloniSchroderVanHandel2024} obtain
norm--moment comparisons for polynomials in free semicircular operators
through ultracontractivity estimates.  Our argument instead uses exact
finite classical multiplication operators.
Our $E_\iid$ is a degree-two matrix polynomial in independent three-point
variables.  An alternative general framework for polynomial random-matrix
concentration is provided by Rajendran and
Tulsiani~\cite{RajendranTulsiani2023}, using recursive matrix Efron--Stein
inequalities.  Their Rademacher bounds reduce to deterministic coefficient
matrices, and Section~8 treats sparse graph matrices with
sparsity-dependent bounds.  We have not checked whether their framework
recovers the SparseStack parameters proved here.

The labelled shared-factor bound underlying Section~\ref{sec:light-sector}
has an equivalent formulation as a maximally entangled-projector sum in
quantum cloning~\cite[Section~IV and Appendix~A]{KayRamanathanKaszlikowski2013}.
We give a self-contained proof and apply it through a hard-core compression
in the band estimates.
The reduction of Section~\ref{sec:convex-transfer} uses the
conditional-expectation and convex-order comparison principle underlying
Hoeffding's sampling comparison~\cite{Hoeffding1963} and its matrix
counterpart due to Gross and Nesme~\cite{GrossNesme2010}; the signed-selector
coupling here compares different laws from their with- and
without-replacement sampling models.

\subsection{Overview of the proof}
\label{sec:overview}

The proof bounds the even moments $\E\tr(E^{2q})$ of the Gram error
$E=U^\T\Pi^\T\Pi U-I_d$ and applies Markov's inequality.  It has four parts.

\emph{Reduction to independent entries} (Section~\ref{sec:convex-transfer}).
Because every column of $\Pi$ has unit norm, $E$ is a hollow sum over pairs
of distinct columns.  A coupling replaces each signed one-hot column
selector $\xi\in\R^b$ of a CountSketch block by a vector $\zeta\in\R^b$ with
independent entries in $\{0,\pm1\}$, in such a way that
$\xi=\E[c_b\zeta\mid\xi]$ for an explicit constant $c_b\le c_\star$.
Conditional Jensen then gives
$\E\tr(E^{2q})\le c_b^{4q}\,\E\tr(E_\iid^{2q})$, where $E_\iid$ is the hollow
Gram error of the independent-entry model.

\emph{A finite operator model} (Section~\ref{sec:fock-model}).  The entries
of the independent model take three values, so multiplication by an entry
acts on a three-dimensional space as the Jacobi matrix \eqref{eq:Jacobi}.
Tensoring over the $mn$ entries, $\E\tr(E_\iid^k)$ becomes the vacuum
matrix element of the $k$-th power of a deterministic operator $\Mcal$ on
$\R^d\otimes(\R^3)^{\otimes mn}$ (Proposition~\ref{prop:vacuum-moment}).
The space is graded by the total level of the tensor factors, and $\Mcal$
splits into three bands $B^+,B^0,B^-$ that raise the grade by $2$, preserve
it, or lower it by $2$.

\emph{Band estimates}
(Sections~\ref{sec:light-sector}--\ref{sec:band-bounds}).  The main
technical result is the uniform bound
\begin{equation}
  \|B^\Delta|_{\Kcal_\nu}\|
  \le(3+\sqrt2)\Bigl[\sqrt{\frac{d+\nu+1}{m}}+\frac{d+\nu+1}{m}+\frac{\nu+1}{s}\Bigr]
  \label{eq:envelope-intro}
\end{equation}
on the grade-$\nu$ sector $\Kcal_\nu$ (Proposition~\ref{prop:band-envelope}).
Each band is a sum of terms that act at two sites of one row of the sketch,
and each such term factors through row operators built from the frame
vectors $u_i$ of that row; its norm is then bounded by the geometric mean of
two Gram operators (Lemma~\ref{lem:block-CS}).  The square-root term in
\eqref{eq:envelope-intro} comes from the creation of two light sites, and
the essential point is that its bound involves $d+\nu$ rather than a product
of $d$ and $\nu$.  This is the content of Lemma~\ref{lem:LSL}: on
configurations with $\ell$ light sites the relevant Gram operator has norm
at most $d+\ell-1$.  The mechanism is that all $\ell$ sites share one
external factor $\R^d$.  Each site contributes the rank-one operator
$|\omega\rangle\langle\omega|$ for a vector $\omega$ of squared norm $d$,
but two such operators attached to different sites overlap only through the
orthogonal projection $UU^\T$, so the Gram matrix of the $\ell$ insertions
has diagonal entries $d$ and off-diagonal blocks of norm at most one.  The term $(\nu+1)/s$ comes from the heavy legs
of the Jacobi matrix, whose coefficient $\sqrt{b-1}$ is large but which act
only at occupied sites; their norms are bounded by counting occupied sites
in a row (Section~\ref{sec:hop}), and the normalization $(b-1)/m\le1/s$
turns the count into the sparsity scale.  Mixed terms are handled by the
same two ingredients and the inequality $2\sqrt{xy}\le x+y$.

\emph{Moment expansion} (Section~\ref{sec:proof}).  Expanding
$\langle\Omega,\Mcal^{2q}\Omega\rangle$ into $3^{2q}$ words in the three
bands, each nonvanishing vacuum-to-vacuum word stays below grade $2q$, which
gives
$\E\tr(E_\iid^{2q})\le d(3\beta_q)^{2q}$ with $\beta_q$ the right side of
\eqref{eq:envelope-intro} at $\nu=2q$.  With $m\gtrsim(d+q)/\varepsilon^2$
and $s\gtrsim q/\varepsilon$ the base $3c_b^2\beta_q/\varepsilon$ is below
$1/2$, and $q=\max\{1,\lceil\log_2(d/\delta)\rceil\}$ makes the failure
probability at most $\delta$.

The theorem has also been formalized in Lean~4; see
Appendix~\ref{app:lean}.

\subsection{Notation and conventions}
\label{sec:notation}

We write $[k]=\{1,\ldots,k\}$.  All sums over $i\ne j$ are over ordered
pairs.  The norm $\|\cdot\|$ is the Euclidean norm of a vector or the
operator norm of a linear map; $\langle\cdot,\cdot\rangle$ is the Euclidean
inner product, extended to tensor products in the usual way.  All spaces are
real and finite-dimensional, tensor products may be read as Kronecker
products, $\dagger$ denotes the adjoint (transpose), and $\preceq$ is the
Loewner order.  For vectors $v,w$ we write $|v\rangle\langle w|$ for the
rank-one map $x\mapsto\langle w,x\rangle v$.  For an operator $\Tcal$ and a
subspace $V$ of its domain, $\Tcal|_V$ is the restriction to $V$ and
$\|\Tcal|_V\|$ its norm.  All probabilities and expectations refer to the
sketching variables.

Throughout, $u_1,\ldots,u_n$ are the rows of a fixed $U$ with $U^\T U=I_d$,
and $U_{ij}=u_iu_j^\T$.  Thus $\|u_i\|\le1$, and for every $S\subset[n]$,
every $x\in\R^d$, and all scalars $(a_i)_{i\in S}$,
\begin{equation}
  \Bigl\|\sum_{i\in S}u_iu_i^\T\Bigr\|\le1,\qquad
  \Bigl\|\sum_{i\in S}a_iu_i\Bigr\|^2\le\sum_{i\in S}a_i^2,\qquad
  \sum_{i\in S}(u_i^\T x)^2\le\|x\|^2 .
  \label{eq:subset-parseval}
\end{equation}
The last two inequalities are $\|U_S^\T\|\le1$ and $\|U_S\|\le1$ for the
row submatrix $U_S$; we call them the synthesis and analysis inequalities.
Every application below involves distinct columns $i$ of a single row $r$
of the sketch.

\section{Reduction to independent entries}
\label{sec:convex-transfer}

\subsection{The hollow Gram error}

For $g\in[s]$ and $i\in[n]$, let $\xi_{gi}=\sigma_{gi}e_{h_g(i)}\in\R^b$ be
the $i$-th column of $H_g$.  Since $\|\xi_{gi}\|=1$ and
\eqref{eq:parseval-frame} holds,
\[
  U^\T\Pi^\T\Pi U
  =\frac1s\sum_{g=1}^s\sum_{i,j=1}^n\langle\xi_{gi},\xi_{gj}\rangle u_iu_j^\T
  =I_d+\frac1s\sum_{g=1}^s\sum_{i\ne j}\langle\xi_{gi},\xi_{gj}\rangle u_iu_j^\T .
\]
Hence the Gram error is the hollow sum
\begin{equation}
  E_\SStack:=U^\T\Pi^\T\Pi U-I_d
  =\frac1s\sum_{g=1}^s\sum_{i\ne j}\langle\xi_{gi},\xi_{gj}\rangle u_iu_j^\T .
  \label{eq:ss-hollow}
\end{equation}
The absence of diagonal terms is what makes the coupling below applicable.

\subsection{A coupling with an independent-entry vector}

Let $\zeta\in\R^b$ have independent coordinates with
\begin{equation}
  \Prob\{\zeta_a=1\}=\Prob\{\zeta_a=-1\}=\frac1{2b},\qquad
  \Prob\{\zeta_a=0\}=1-\frac1b\qquad(a\in[b]),
  \label{eq:zeta-law}
\end{equation}
and put
\[
  p_0=\Prob\{\zeta=0\}=\Bigl(1-\frac1b\Bigr)^b,\qquad c_b=(1-p_0)^{-1},
\]
so that $c_1=1$.  Given $\zeta$, define a signed basis vector $\xi$ as
follows.  If $\zeta\ne0$, choose an index $\hat a$ uniformly from
$\supp\zeta=\{a:\zeta_a\ne0\}$, independently of everything else, and set
$\xi=\zeta_{\hat a}e_{\hat a}$.  If $\zeta=0$, choose $\hat a$ uniformly
from $[b]$ and an independent Rademacher sign $\sigma$, and set
$\xi=\sigma e_{\hat a}$.

\begin{lemma}[Selector coupling]
\label{lem:coupling}
\begin{enumerate}[label=\textup{(\roman*)}]
\item $\xi$ is uniformly distributed on the $2b$ signed basis vectors of
  $\R^b$ and is independent of $|\supp\zeta|$.
\item $\E[\zeta\mid\xi]=c_b^{-1}\xi$; equivalently, $\xi=\E[c_b\zeta\mid\xi]$.
\item $1\le c_b\le c_\star$.
\end{enumerate}
\end{lemma}

\begin{proof}
(i) Fix $k\in[b]$ and $\sigma\in\{-1,1\}$.  Conditionally on
$|\supp\zeta|=k_0\ge1$, the support is a uniformly random $k_0$-subset of
$[b]$ and the signs on the support are independent Rademacher variables, so
\[
  \Prob\{\xi=\sigma e_k\mid|\supp\zeta|=k_0\}
  =\frac{k_0}{b}\cdot\frac1{k_0}\cdot\frac12=\frac1{2b},
\]
the three factors being the probabilities that $k\in\supp\zeta$, that
$\hat a=k$ given $k\in\supp\zeta$, and that $\zeta_k=\sigma$.  The fallback
rule gives the same value when $k_0=0$.  Thus the conditional law of $\xi$
given $|\supp\zeta|$ is uniform on the $2b$ signed basis vectors, which is
(i).

(ii) Fix $\xi=\sigma e_k$.  On $\{\zeta\ne0\}$ the construction selected
the coordinate $k$, so $\zeta_k=\sigma$; on $\{\zeta=0\}$, $\zeta_k=0$.  By
(i), $\Prob\{\zeta\ne0\mid\xi=\sigma e_k\}=1-p_0$, and therefore
$\E[\zeta_k\mid\xi=\sigma e_k]=(1-p_0)\sigma$.  For $a\ne k$, flipping the
sign of $\zeta_a$ preserves the joint law of $\zeta$ and the auxiliary
randomness, preserves $\supp\zeta$ and the selected index, and hence
preserves $\xi$; so the conditional law of $\zeta_a$ given $\xi=\sigma e_k$
is symmetric and $\E[\zeta_a\mid\xi=\sigma e_k]=0$.  Together,
$\E[\zeta\mid\xi]=(1-p_0)\xi=c_b^{-1}\xi$.

(iii) $0\le p_0\le e^{-1}$, because $(1-1/b)^b\le e^{-1}$.
\end{proof}

\subsection{Convex-order transfer}

Take independent copies $(\xi_{gi},\zeta_{gi})$ of the coupled pair, one
for each $(g,i)\in[s]\times[n]$.  By Lemma~\ref{lem:coupling}(i) the
vectors $\xi_{gi}$ have the joint law of the columns of the blocks $H_g$,
so we may realize $E_\SStack$ through them.  Define the independent-entry
counterpart
\begin{equation}
  E_\iid=\frac1s\sum_{g=1}^s\sum_{i\ne j}
    \langle\zeta_{gi},\zeta_{gj}\rangle u_iu_j^\T ,
  \label{eq:iid-vector}
\end{equation}
and let $\Gcal_\xi$ be the $\sigma$-algebra generated by all $\xi_{gi}$.
Conditionally on $\Gcal_\xi$, the vectors $\zeta_{gi}$ are independent, and
the conditional law of $\zeta_{gi}$ depends only on $\xi_{gi}$.  Hence, for
$i\ne j$, Lemma~\ref{lem:coupling}(ii) gives
\[
  \E[\langle\zeta_{gi},\zeta_{gj}\rangle\mid\Gcal_\xi]
  =\bigl\langle\E[\zeta_{gi}\mid\xi_{gi}],\E[\zeta_{gj}\mid\xi_{gj}]\bigr\rangle
  =c_b^{-2}\langle\xi_{gi},\xi_{gj}\rangle ,
\]
and summing against the deterministic matrices $u_iu_j^\T$,
\begin{equation}
  E_\SStack=\E[c_b^2E_\iid\mid\Gcal_\xi].
  \label{eq:conditional-transfer}
\end{equation}
The restriction to $i\ne j$ is essential: by Lemma~\ref{lem:coupling}(i),
$\E[\|\zeta_{gi}\|^2\mid\xi_{gi}]=\E|\supp\zeta_{gi}|=1$, whereas
$\|\xi_{gi}\|^2=1$, so a diagonal term would pick up the factor
$c_b^2\ne1$ when $b>1$.

\begin{proposition}[Moment transfer]
\label{prop:transfer}
For every integer $q\ge1$,
\begin{equation}
  \E\tr(E_\SStack^{2q})\le c_b^{4q}\,\E\tr(E_\iid^{2q}).
  \label{eq:convex-moment}
\end{equation}
\end{proposition}

\begin{proof}
Both $E_\SStack$ and $E_\iid$ are real symmetric, since the $(j,i)$ term of
each ordered sum is the transpose of the $(i,j)$ term.  On real symmetric
matrices the function $f(A)=\tr(A^{2q})=\|A\|_{S_{2q}}^{2q}$ is convex,
being the $2q$-th power of the Schatten norm.  Conditional Jensen and
\eqref{eq:conditional-transfer} give
$f(E_\SStack)=f(\E[c_b^2E_\iid\mid\Gcal_\xi])
\le\E[f(c_b^2E_\iid)\mid\Gcal_\xi]$, and taking expectations yields
$\E f(E_\SStack)\le\E f(c_b^2E_\iid)=c_b^{4q}\,\E f(E_\iid)$.
\end{proof}

Comparisons of this kind, in which a dependent model is replaced by an
independent one through a conditional barycenter and Jensen's inequality,
go back to Hoeffding's treatment of sampling without
replacement~\cite{Hoeffding1963}; the matrix version is due to Gross and
Nesme~\cite{GrossNesme2010} and is used, for example, in~\cite{Tropp2011}.
See~\cite{ShakedShanthikumar2007} for the convex order in general.  The
specific ingredient here is the signed-selector coupling of
Lemma~\ref{lem:coupling}, whose target law admits the representation of
the next section; the sampling laws differ from those in the cited
with- and without-replacement comparisons.

\section{A finite tensor-product model for the moments}
\label{sec:fock-model}

This section expresses $\E\tr(E_\iid^k)$ as a matrix element of a
deterministic operator on a finite tensor product.  The construction is
elementary linear algebra on a finite probability space; the words
``site,'' ``level,'' and ``vacuum'' name index sets and basis vectors and
carry no further meaning.

\subsection{Independent three-point entries}

Identify the pair $(g,a)\in[s]\times[b]$ with the row $r=(g-1)b+a\in[m]$ of
$\Pi$, and set
\[
  \eta_{ri}=\sqrt b\,\zeta_{gi}(a)\qquad((r,i)\in[m]\times[n]).
\]
The $mn$ variables $\eta_{ri}$ are independent, each with the law
\begin{equation}
  \Prob\{\eta=\sqrt b\}=\Prob\{\eta=-\sqrt b\}=\frac1{2b},\qquad
  \Prob\{\eta=0\}=1-\frac1b ,
  \label{eq:eta-law}
\end{equation}
so that $\E\eta=0$, $\E\eta^2=1$, and $\E\eta^4=b$.  Expanding the inner
products in \eqref{eq:iid-vector} and using $m=sb$,
\begin{equation}
  E_\iid=\frac1m\sum_{r=1}^m\sum_{i\ne j}\eta_{ri}\eta_{rj}\,u_iu_j^\T .
  \label{eq:iid-row}
\end{equation}
Under the parameter choice \eqref{eq:integer-parameters} one has $b\ge2$
(see \eqref{eq:b-at-least-two} below), so $\eta$ takes each of its three
values with positive probability.  We assume $b\ge2$ from now on.

\subsection{Multiplication by one entry}

Let $L_2(\eta)$ be the space of real functions on the three-point support
of \eqref{eq:eta-law}, with inner product
$\langle f,g\rangle=\E[f(\eta)g(\eta)]$.  It is three-dimensional, and the
functions
\begin{equation}
  e_0=1,\qquad e_1=\eta,\qquad e_2=\frac{\eta^2-1}{\sqrt{b-1}}
  \label{eq:one-site-basis}
\end{equation}
form an orthonormal basis: the cross inner products vanish by centering and
symmetry, $\E\eta^2=1$, and $\E(\eta^2-1)^2=\E\eta^4-2\E\eta^2+1=b-1$.
Since $\eta^3=b\eta$ on the support, multiplication by $\eta$ acts by
\[
  \eta e_0=e_1,\qquad
  \eta e_1=\eta^2=e_0+\sqrt{b-1}\,e_2,\qquad
  \eta e_2=\frac{\eta^3-\eta}{\sqrt{b-1}}=\sqrt{b-1}\,e_1 ,
\]
so its matrix in the basis $(e_0,e_1,e_2)$ is the Jacobi matrix
\begin{equation}
  J=\begin{pmatrix}0&1&0\\1&0&\sqrt{b-1}\\0&\sqrt{b-1}&0\end{pmatrix}.
  \label{eq:Jacobi}
\end{equation}
Write $|a\rangle\langle c|$ for the matrix unit that sends $e_c$ to $e_a$
and annihilates the other two basis vectors, and define the four
\emph{legs}
\begin{equation}
  P^\dagger=|1\rangle\langle0|,\qquad P=|0\rangle\langle1|,\qquad
  R=|2\rangle\langle1|,\qquad R^\dagger=|1\rangle\langle2| ,
  \label{eq:local-legs}
\end{equation}
so that
\begin{equation}
  J=P+P^\dagger+\sqrt{b-1}\,(R+R^\dagger).
  \label{eq:J-decomp}
\end{equation}
The nonzero actions of the legs on the basis, and their coefficients in
\eqref{eq:J-decomp}, are as follows; all other actions on basis vectors are
zero.
\begin{center}
\renewcommand{\arraystretch}{1.15}
\begin{tabular}{cccc}
\toprule
Leg & Nonzero action & Change of index & Coefficient in $J$\\
\midrule
$P^\dagger$ & $e_0\mapsto e_1$ & $+1$ & $1$\\
$P$         & $e_1\mapsto e_0$ & $-1$ & $1$\\
$R$         & $e_1\mapsto e_2$ & $+1$ & $\sqrt{b-1}$\\
$R^\dagger$ & $e_2\mapsto e_1$ & $-1$ & $\sqrt{b-1}$\\
\bottomrule
\end{tabular}
\end{center}
Thus $P^\dagger$ and $R$ raise the basis index and $P$ and $R^\dagger$
lower it.  We call $P,P^\dagger$ the \emph{light} legs
and $R,R^\dagger$ the \emph{heavy} legs; the names refer to the coefficient
$\sqrt{b-1}$ in \eqref{eq:J-decomp}, which is the only place where the
sparsity enters the local model.  The representation of a finitely
supported random variable by a Jacobi matrix in an orthogonal-polynomial
basis, and the decomposition of that matrix into raising and lowering parts
(the diagonal part vanishes here because $\eta$ is symmetric), is the
finite case of the interacting Fock space construction of Accardi and
Bo\.zejko~\cite{AccardiBozejko1998}; the three-level case used here is
completely described by \eqref{eq:Jacobi}--\eqref{eq:J-decomp}.

\subsection{Sites, patterns, and the product space}

Let $\sfI=[m]\times[n]$; its elements $(r,i)$ are called \emph{sites}.  Let
\begin{equation}
  \Fcal=\bigotimes_{(r,i)\in\sfI}\R^3=(\R^3)^{\otimes mn},
  \label{eq:Fock-product}
\end{equation}
with the factors ordered lexicographically, so that the site $(r,i)$
occupies position $\operatorname{pos}(r,i)=(r-1)n+i$.  For a $3\times3$
matrix $Z$ and a site $(r,i)$, the \emph{lift} of $Z$ to the site $(r,i)$
is the $3^{mn}\times3^{mn}$ matrix
\begin{equation}
  Z_{ri}=I_3^{\otimes(\operatorname{pos}(r,i)-1)}\otimes Z\otimes
         I_3^{\otimes(mn-\operatorname{pos}(r,i))},
  \label{eq:site-lift}
\end{equation}
which applies $Z$ to the factor indexed by $(r,i)$ and the identity to all
other factors.  Lifts to distinct sites commute, and
$(Z_{ri})^\dagger=(Z^\dagger)_{ri}$.  We write
$J_{ri},P_{ri},P^\dagger_{ri},R_{ri},R^\dagger_{ri}$ for the lifts of the
corresponding local matrices.

A \emph{pattern} is a map $\tau:\sfI\to\{0,1,2\}$; the value $\tau_{ri}$ is
the \emph{level} of the site $(r,i)$.  The pattern $\tau$ indexes the basis
vector
\begin{equation}
  |\tau\rangle=\bigotimes_{(r,i)\in\sfI}e_{\tau_{ri}}\in\Fcal ,
  \label{eq:pattern-basis}
\end{equation}
and these $3^{mn}$ vectors form an orthonormal basis of $\Fcal$.  The
all-zero pattern gives the \emph{vacuum} $\Omega=e_0^{\otimes mn}$.  A site
at level $0$, $1$, or $2$ in $\tau$ is called \emph{free}, \emph{light}, or
\emph{heavy}, respectively, and we write
\begin{equation}
  S_j(\tau)=\{(r,i)\in\sfI:\tau_{ri}=j\},\qquad
  S_j^r(\tau)=\{i\in[n]:\tau_{ri}=j\},\qquad
  F_r(\tau)=S_0^r(\tau)
  \label{eq:pattern-sets}
\end{equation}
for the set of level-$j$ sites, the level-$j$ columns of row $r$, and the
free columns of row $r$.  For a fixed row $r$, the pattern
$\tau[a\mapsto j]$ agrees with $\tau$ except that the site $(r,a)$ is
assigned level $j$, and $\tau[a\mapsto j,\,c\mapsto k]$ is defined
analogously for two sites of the same row.  In this notation, the lift of a
leg $Z=|a'\rangle\langle a|$ acts on basis vectors by
\begin{equation}
  Z_{ri}|\tau\rangle=
  \begin{cases}
    |\tau[i\mapsto a']\rangle,&\tau_{ri}=a,\\
    0,&\text{otherwise},
  \end{cases}
  \label{eq:leg-on-pattern}
\end{equation}
by \eqref{eq:site-lift} and \eqref{eq:pattern-basis}.

The link with probability is the following.  By independence, the products
$\prod_{(r,i)}e_{\tau_{ri}}(\eta_{ri})$ form an orthonormal basis of the
$L_2$ space of the array $(\eta_{ri})$, so the map $\Ucal$ defined by
$\Ucal|\tau\rangle=\prod_{(r,i)}e_{\tau_{ri}}(\eta_{ri})$ is an isometry
from $\Fcal$ onto that space.  Under $\Ucal$, the vacuum corresponds to the
constant function $1$ and $J_{ri}$ to multiplication by $\eta_{ri}$.
Consequently, for every polynomial $p$ in the commuting operators $J_{ri}$,
\begin{equation}
  \bigl\langle\Omega,p\bigl((J_{ri})\bigr)\Omega\bigr\rangle
  =\E\,p\bigl((\eta_{ri})\bigr).
  \label{eq:scalar-vacuum-rule}
\end{equation}

\subsection{The moment identity}

Let $q_1,\ldots,q_d$ be the standard basis of $\R^d$, and define
\begin{equation}
  \Mcal=\frac1m\sum_{r=1}^m\sum_{i\ne j}U_{ij}\otimes J_{ri}J_{rj}
  \qquad\text{on }\R^d\otimes\Fcal .
  \label{eq:M}
\end{equation}
We refer to the factor $\R^d$ as the \emph{external factor}.  The operator
$\Mcal$ is self-adjoint: $(U_{ij}\otimes J_{ri}J_{rj})^\dagger
=U_{ji}\otimes J_{rj}J_{ri}$, the two local factors commute because
$i\ne j$, and the ordered sum is invariant under $(i,j)\mapsto(j,i)$.

\begin{proposition}[Vacuum-moment identity]
\label{prop:vacuum-moment}
For every integer $k\ge0$,
\begin{equation}
  \E\tr(E_\iid^k)
  =\sum_{\alpha=1}^d\bigl\langle q_\alpha\otimes\Omega,\,
      \Mcal^k(q_\alpha\otimes\Omega)\bigr\rangle .
  \label{eq:vacuum-moment}
\end{equation}
\end{proposition}

\begin{proof}
Under the identification $I_d\otimes\Ucal$, the operator $\Mcal$ becomes
multiplication of $\R^d$-valued functions of the array $(\eta_{ri})$ by the
matrix-valued function $E_\iid$ of \eqref{eq:iid-row}, and $\Mcal^k$
becomes multiplication by $E_\iid^k$.  The vector $q_\alpha\otimes\Omega$
corresponds to the constant function $q_\alpha$, so the $\alpha$-th summand
in \eqref{eq:vacuum-moment} equals $\E[q_\alpha^\T E_\iid^kq_\alpha]$, and
summing over $\alpha$ gives $\E\tr(E_\iid^k)$.

The same identity can be verified term by term.  Expanding the $k$-th power
of \eqref{eq:M},
\[
  \Mcal^k=\frac1{m^k}
  \sum_{\substack{r_1,\ldots,r_k\in[m]\\ i_t\ne j_t\ (1\le t\le k)}}
  U_{i_1j_1}\cdots U_{i_kj_k}\otimes\prod_{t=1}^kJ_{r_ti_t}J_{r_tj_t},
\]
and \eqref{eq:scalar-vacuum-rule} gives
\[
  \bigl\langle q_\alpha\otimes\Omega,\Mcal^k(q_\alpha\otimes\Omega)\bigr\rangle
  =\frac1{m^k}
  \sum_{\substack{r_1,\ldots,r_k\\ i_t\ne j_t}}
  q_\alpha^\T U_{i_1j_1}\cdots U_{i_kj_k}q_\alpha\;
  \E\prod_{t=1}^k\eta_{r_ti_t}\eta_{r_tj_t},
\]
which is exactly the expansion of $\E[q_\alpha^\T E_\iid^kq_\alpha]$
obtained by substituting \eqref{eq:iid-row} $k$ times.
\end{proof}

\begin{example}[The case $k=2$]
\label{ex:k2}
Expanding \eqref{eq:iid-row} twice,
\[
  \E\tr(E_\iid^2)
  =\frac1{m^2}\sum_{r_1,r_2=1}^m
   \sum_{\substack{i_1\ne j_1\\ i_2\ne j_2}}
   \E\bigl[\eta_{r_1i_1}\eta_{r_1j_1}\eta_{r_2i_2}\eta_{r_2j_2}\bigr]
   \tr(U_{i_1j_1}U_{i_2j_2}).
\]
Independence factors the expectation over distinct sites, and a site that
occurs once contributes $\E\eta=0$; since $i_t\ne j_t$, each factor
$\eta_{r_ti_t}\eta_{r_tj_t}$ involves two distinct sites, so every site
occurs twice exactly when $r_1=r_2$ and $\{i_1,j_1\}=\{i_2,j_2\}$, and
then the expectation is $\E\eta_{ri}^2\eta_{rj}^2=1$.  For each ordered
pair $(i_1,j_1)$ there are two admissible ordered pairs $(i_2,j_2)$, namely
$(i_1,j_1)$ and $(j_1,i_1)$, and the $m$ choices of the common row cancel
one factor $1/m$.  Hence
\[
  \E\tr(E_\iid^2)
  =\frac1m\sum_{i\ne j}\bigl[\tr(U_{ij}U_{ij})+\tr(U_{ij}U_{ji})\bigr]
  =\frac1m\sum_{i\ne j}\bigl[(u_i^\T u_j)^2+\|u_i\|^2\|u_j\|^2\bigr].
\]
On the operator side, $J_{r_2i_2}J_{r_2j_2}\Omega$ is the basis vector with
light sites $(r_2,i_2)$ and $(r_2,j_2)$, since only $P^\dagger$ acts
nontrivially on $e_0$.  The inner product of $\Omega$ with
$J_{r_1i_1}J_{r_1j_1}$ applied to that vector is nonzero only if the second
pair of legs returns both sites to level $0$, which forces $r_1=r_2$ and
$\{i_1,j_1\}=\{i_2,j_2\}$; thus
\[
  \bigl\langle\Omega,J_{r_1i_1}J_{r_1j_1}J_{r_2i_2}J_{r_2j_2}\Omega\bigr\rangle
  =\1_{\{r_1=r_2,\ \{i_1,j_1\}=\{i_2,j_2\}\}},
\]
and substituting into \eqref{eq:vacuum-moment} reproduces the same sum.
The vacuum matrix element is bookkeeping for the surviving terms of the
classical moment expansion.
\end{example}

\subsection{Grading and the three bands}

The \emph{grade} of a pattern is
\begin{equation}
  \gr(\tau)=\sum_{(r,i)\in\sfI}\tau_{ri}=|S_1(\tau)|+2|S_2(\tau)| .
  \label{eq:pattern-grade}
\end{equation}
For an integer $\nu$ let
\begin{equation}
  \Hcal_\nu=\operatorname{span}\{|\tau\rangle:\gr(\tau)=\nu\}\subset\Fcal,
  \qquad
  \Kcal_\nu=\R^d\otimes\Hcal_\nu ,
  \label{eq:grade-sectors}
\end{equation}
so that $\Fcal=\bigoplus_{\nu=0}^{2mn}\Hcal_\nu$ orthogonally, and let
$\pi_\nu$ and $\Pcal_\nu=I_d\otimes\pi_\nu$ be the orthogonal projections
onto $\Hcal_\nu$ and $\Kcal_\nu$ (both zero if $\nu\notin\{0,\ldots,2mn\}$).

Each leg changes the level of one site by $\pm1$ and hence the grade by
$\pm1$.  Expanding $J_{ri}J_{rj}$ by \eqref{eq:J-decomp} expresses each
summand of $\Mcal$ as a sum of sixteen terms $c\,U_{ij}\otimes Z_{ri}W_{rj}$
with $Z,W\in\{P,P^\dagger,R,R^\dagger\}$ and $c\in\{1,\sqrt{b-1},b-1\}$,
each of which maps $\Kcal_\nu$ into $\Kcal_{\nu+2}$, $\Kcal_\nu$, or
$\Kcal_{\nu-2}$.  Consequently $\Pcal_\mu\Mcal\Pcal_\nu=0$ unless
$\mu-\nu\in\{-2,0,2\}$, and
\begin{equation}
  \Mcal=B^++B^0+B^-,\qquad
  B^\Delta:=\sum_{\nu}\Pcal_{\nu+2\Delta}\Mcal\Pcal_\nu
  \quad(\Delta\in\{1,0,-1\}),
  \label{eq:bands}
\end{equation}
where $B^\Delta$ maps $\Kcal_\nu$ into $\Kcal_{\nu+2\Delta}$.  Since
$\Mcal$ is self-adjoint and the $\Pcal_\nu$ are orthogonal projections,
\[
  (B^+)^\dagger=\sum_\nu(\Pcal_{\nu+2}\Mcal\Pcal_\nu)^\dagger
  =\sum_\nu\Pcal_\nu\Mcal\Pcal_{\nu+2}
  =\sum_\mu\Pcal_{\mu-2}\Mcal\Pcal_\mu=B^-,
\]
after the reindexing $\mu=\nu+2$; similarly $(B^0)^\dagger=B^0$.

\subsection{Band inventory}
\label{sec:band-inventory}

Put $h=b-1$.  In this subsection every sum is over $r\in[m]$ and ordered
pairs $i\ne j$.  By \eqref{eq:leg-on-pattern}, a term
$U_{ij}\otimes Z_{ri}W_{rj}$ with legs $Z=|a'\rangle\langle a|$ and
$W=|c'\rangle\langle c|$ acts on $x\otimes|\tau\rangle$, $x\in\R^d$, by
\begin{equation}
  (U_{ij}\otimes Z_{ri}W_{rj})(x\otimes|\tau\rangle)
  =\begin{cases}
    u_i\,(u_j^\T x)\otimes|\tau[i\mapsto a',\,j\mapsto c']\rangle,
      &\tau_{ri}=a\text{ and }\tau_{rj}=c,\\
    0,&\text{otherwise:}
  \end{cases}
  \label{eq:term-action}
\end{equation}
the external vector is contracted against $u_j$ and re-emitted along $u_i$,
while the two sites change level.  Define the light--light and heavy--heavy
families
\begin{equation}
\begin{aligned}
  L_+&=\sum U_{ij}\otimes P^\dagger_{ri}P^\dagger_{rj},&
  L_0&=\sum U_{ij}\otimes(P^\dagger_{ri}P_{rj}+P_{ri}P^\dagger_{rj}),&
  L_-&=L_+^\dagger,\\
  H_+&=\sum U_{ij}\otimes R_{ri}R_{rj},&
  H_0&=\sum U_{ij}\otimes(R_{ri}R^\dagger_{rj}+R^\dagger_{ri}R_{rj}),&
  H_-&=H_+^\dagger,
\end{aligned}
\label{eq:light-heavy-defs}
\end{equation}
and the mixed families
\begin{equation}
\begin{aligned}
  X_+&=\sum U_{ij}\otimes P^\dagger_{ri}R_{rj},&
  Y_+&=\sum U_{ij}\otimes R_{ri}P^\dagger_{rj},\\
  X_0&=\sum U_{ij}\otimes P^\dagger_{ri}R^\dagger_{rj},&
  Y_0&=\sum U_{ij}\otimes R^\dagger_{ri}P^\dagger_{rj}.
\end{aligned}
\label{eq:mixed-defs}
\end{equation}
For any two legs $Z,W$, using $U_{ij}^\dagger=U_{ji}$, commuting the
factors at distinct sites, and exchanging the dummy indices $i$ and $j$,
\begin{equation}
  \Bigl(\sum U_{ij}\otimes Z_{ri}W_{rj}\Bigr)^\dagger
  =\sum U_{ij}\otimes W^\dagger_{ri}Z^\dagger_{rj} ;
  \label{eq:family-adjoint}
\end{equation}
thus, for instance, $X_+^\dagger=\sum U_{ij}\otimes R^\dagger_{ri}P_{rj}$
and $Y_0^\dagger=\sum U_{ij}\otimes P_{ri}R_{rj}$.  Sorting the sixteen
terms of the expansion of $J_{ri}J_{rj}$ by grade shift, and recording
that a term with zero, one, or two heavy legs carries the coefficient $1$,
$\sqrt h$, or $h$, gives the decomposition
\begin{equation}
\begin{aligned}
  B^+&=\frac1m\bigl[L_++\sqrt h\,(X_++Y_+)+hH_+\bigr],\\
  B^0&=\frac1m\bigl[L_0+\sqrt h\,(X_0+X_0^\dagger+Y_0+Y_0^\dagger)+hH_0\bigr],\\
  B^-&=\frac1m\bigl[L_-+\sqrt h\,(X_+^\dagger+Y_+^\dagger)+hH_-\bigr].
\end{aligned}
\label{eq:exact-inventory}
\end{equation}

\subsection{Row operators and a factorization identity}
\label{sec:row-operators}

For a $3\times3$ matrix $Z$ and a row $r$, define
\begin{equation}
  \Scal_r(Z)=\sum_{i=1}^n u_i\otimes Z_{ri}:\Fcal\to\R^d\otimes\Fcal,
  \qquad
  \Scal_r(Z)^\dagger=\sum_{i=1}^n u_i^\T\otimes(Z^\dagger)_{ri}:
  \R^d\otimes\Fcal\to\Fcal ,
  \label{eq:row-operator}
\end{equation}
where $u_i$ is regarded as the map $\R\to\R^d$, $t\mapsto tu_i$, and
$u_i^\T$ as its adjoint.  Thus $\Scal_r(Z)$ applies $Z$ at one site of row
$r$ and synthesizes an external vector from the frame vectors of that row,
while $\Scal_r(Z)^\dagger$ contracts the external vector against those
frame vectors and applies $Z^\dagger$.  The associated \emph{Gram operator}
\begin{equation}
  \Gcal(Z)=\sum_{r=1}^m\Scal_r(Z)\Scal_r(Z)^\dagger
  =\sum_{r=1}^m\sum_{i,j=1}^nU_{ij}\otimes Z_{ri}(Z^\dagger)_{rj}
  \label{eq:gram-operator}
\end{equation}
is positive semidefinite on $\R^d\otimes\Fcal$.  If $Z$ is
grade-homogeneous, in particular if $Z$ is one of the four legs, then
$\Scal_r(Z)$ and $\Scal_r(Z)^\dagger$ shift the grade by opposite amounts
and $\Gcal(Z)$ preserves grade.

\begin{lemma}[Factorization identity]
\label{lem:factorization}
For all $3\times3$ matrices $Z,W$,
\begin{equation}
  \sum_{r=1}^m\sum_{i\ne j}U_{ij}\otimes Z_{ri}W_{rj}
  =\sum_{r=1}^m\Scal_r(Z)\Scal_r(W^\dagger)^\dagger
   -\sum_{r=1}^m\sum_{i=1}^nU_{ii}\otimes(ZW)_{ri}.
  \label{eq:factorization}
\end{equation}
\end{lemma}

\begin{proof}
By \eqref{eq:row-operator},
$\Scal_r(Z)\Scal_r(W^\dagger)^\dagger=\sum_{i,j}u_iu_j^\T\otimes Z_{ri}W_{rj}$,
and the terms with $i=j$ are $U_{ii}\otimes(ZW)_{ri}$.
\end{proof}

The local products that occur below are
\begin{equation}
\begin{gathered}
  (P^\dagger)^2=R^2=P^\dagger R=P^\dagger R^\dagger=0,\\
  P^\dagger P=R^\dagger R=|1\rangle\langle1|,\qquad
  PP^\dagger=|0\rangle\langle0|,\qquad
  RR^\dagger=|2\rangle\langle2| .
\end{gathered}
  \label{eq:local-products}
\end{equation}
Define the diagonal operators
\begin{equation}
  D_0=\sum_{r,i}U_{ii}\otimes|0\rangle\langle0|_{ri},\qquad
  D_1=\sum_{r,i}U_{ii}\otimes|1\rangle\langle1|_{ri},\qquad
  D_H=\sum_{r,i}U_{ii}\otimes\bigl(|1\rangle\langle1|+|2\rangle\langle2|\bigr)_{ri},
  \label{eq:D-ops}
\end{equation}
the light Gram operator, and the transpose light hop
\begin{equation}
  \Ghat=\Gcal(P^\dagger)=\sum_{r}\sum_{i,j}U_{ij}\otimes P^\dagger_{ri}P_{rj},
  \qquad
  H'=\sum_r\sum_{i\ne j}U_{ij}\otimes P_{ri}P^\dagger_{rj}.
  \label{eq:Ghat}
\end{equation}
By \eqref{eq:family-adjoint}, $H'$ is self-adjoint.

\begin{corollary}
\label{cor:factorization}
The following identities hold:
\begin{align}
  L_+&=\sum_r\Scal_r(P^\dagger)\Scal_r(P)^\dagger,&
  H_+&=\sum_r\Scal_r(R)\Scal_r(R^\dagger)^\dagger,
  \label{eq:fact-plus}\\
  X_+&=\sum_r\Scal_r(P^\dagger)\Scal_r(R^\dagger)^\dagger,&
  X_0&=\sum_r\Scal_r(P^\dagger)\Scal_r(R)^\dagger,
  \label{eq:fact-mixed}\\
  L_0&=(\Ghat-D_1)+H',&
  \Gcal(P)&=D_0+H',
  \label{eq:fact-light}\\
  H_0&=\Gcal(R)+\Gcal(R^\dagger)-D_H .
  \label{eq:fact-heavy}
\end{align}
\end{corollary}

\begin{proof}
Apply Lemma~\ref{lem:factorization} with $(Z,W)$ equal to
$(P^\dagger,P^\dagger)$, $(R,R)$, $(P^\dagger,R)$, and
$(P^\dagger,R^\dagger)$; by \eqref{eq:local-products} the diagonal term
vanishes in each case, which gives \eqref{eq:fact-plus} and
\eqref{eq:fact-mixed}.  With $(Z,W)=(P^\dagger,P)$ the left side of
\eqref{eq:factorization} is the first half of $L_0$ and the right side is
$\Ghat-D_1$; adding $H'$ gives the first identity in \eqref{eq:fact-light}.
With $(Z,W)=(P,P^\dagger)$ the left side is $H'$ and the right side is
$\Gcal(P)-D_0$.  Finally, $(Z,W)=(R,R^\dagger)$ and $(R^\dagger,R)$ express
the two halves of $H_0$ as $\Gcal(R)-\sum U_{ii}\otimes|2\rangle\langle2|_{ri}$
and $\Gcal(R^\dagger)-\sum U_{ii}\otimes|1\rangle\langle1|_{ri}$.
\end{proof}

The families $Y_+$ and $Y_0$ are not treated by factorization.  For $Y_+$
the diagonal term in \eqref{eq:factorization} would be
$\sum U_{ii}\otimes(RP^\dagger)_{ri}$ with $RP^\dagger=|2\rangle\langle0|$, a
same-site transition from level $0$ to level $2$ that does not occur in
$\Mcal$; for $Y_0$ the diagonal term vanishes, but the resulting
Cauchy--Schwarz bound would involve $\Gcal(P)$, whose norm is of order $m$.
Both families are bounded directly in Section~\ref{sec:hop}, using the fact
that a heavy leg can act only at an occupied site of the same row.

\section{The shared-factor light-sector inequality}
\label{sec:light-sector}

The operator $\Ghat=\Gcal(P^\dagger)$ moves one light site of a row to a
free site of the same row, or leaves it in place through the terms $i=j$,
with the external action $U_{ij}$.  It preserves the set of heavy sites and
the number of light sites.  For a set $T\subset\sfI$ and an integer
$\ell\ge0$, write $|S;T\rangle$ for the basis vector with light sites $S$
and heavy sites $T$, and let
\begin{equation}
  F_T=\sfI\setminus T,\qquad
  \sfK_{T,\ell}=\operatorname{span}\{|S;T\rangle:S\subset F_T,\ |S|=\ell\}.
  \label{eq:fixed-light-block}
\end{equation}
We call $\sfK_{T,\ell}$ a \emph{hard-core} block: its basis vectors are
indexed by sets of $\ell$ distinct light sites, so no site is occupied
twice, in contrast with the labelled tensor product introduced in the proof
below, in which repetitions are allowed.  Each $\R^d\otimes\sfK_{T,\ell}$ is
invariant under $\Ghat$, and
$\Hcal_\nu=\bigoplus_{T:\,2|T|\le\nu}\sfK_{T,\nu-2|T|}$.

\begin{lemma}[Light-sector lemma]
\label{lem:LSL}
On $\R^d\otimes\sfK_{T,\ell}$, $\Ghat=0$ if $\ell=0$, and
\begin{equation}
  0\preceq\Ghat\preceq(d+\ell-1)\,I\qquad\text{if }\ell\ge1 .
  \label{eq:LSL-block}
\end{equation}
Consequently, for every $\kappa\ge0$,
\begin{equation}
  0\preceq\Ghat\preceq(d+\kappa-1)\,I\preceq(d+\kappa)\,I
  \qquad\text{on }\Kcal_\kappa .
  \label{eq:LSL-grade}
\end{equation}
\end{lemma}

The point of the lemma is that the bound is additive in $d$ and $\ell$.
Bounding the contribution of each of the $\ell$ light sites separately
would give $d\ell$.  The improvement comes from the fact that all light
sites share a single external factor $\R^d$: each site contributes the
rank-one operator $|\omega\rangle\langle\omega|$ for a vector $\omega$ of
squared norm $d$, but two such operators attached to different sites
overlap only through $UU^\T$, an operator of norm one.

\begin{proof}
If $\ell=0$, every $P_{rj}$ annihilates $\sfK_{T,0}$ and $\Ghat=0$.  Assume
$\ell\ge1$.  The proof has three steps: we bound a sum of $\ell$ copies of
a one-particle operator acting on different slots of a labelled tensor
product, identify $\Ghat$ as a compression of that sum, and conclude.

\medskip\noindent\emph{Step 1: a labelled model.}
Let $\sfA=\R^d$, $\sfR=\R^m$, $\sfC=\R^n$, and $\sfX=\sfR\otimes\sfC$; in
this proof $e_r\in\sfR$ and $e_i\in\sfC$ denote standard basis vectors.  For
each row $r$ put
\begin{equation}
  w_r=\sum_{i=1}^nu_i\otimes e_r\otimes e_i\in\sfA\otimes\sfX,
  \qquad
  Q=\sum_{r=1}^m|w_r\rangle\langle w_r| .
  \label{eq:Q-shared}
\end{equation}
Since $\langle w_r,w_{r'}\rangle=\delta_{rr'}\sum_i\|u_i\|^2=d\,\delta_{rr'}$,
the vectors $w_r$ are orthogonal with squared norm $d$; up to a permutation
of tensor factors, $Q=I_\sfR\otimes|\omega\rangle\langle\omega|$ with
$\omega=\sum_iu_i\otimes e_i\in\sfA\otimes\sfC$.  On basis vectors,
\begin{equation}
  Q\bigl(x\otimes e_r\otimes e_j\bigr)=\sum_{i=1}^nU_{ij}x\otimes e_r\otimes e_i
  \qquad(x\in\sfA),
  \label{eq:Q-action}
\end{equation}
because only $w_r$ has a nonzero inner product with
$x\otimes e_r\otimes e_j$, namely $\langle u_j,x\rangle$.

Let $\sfX_1,\ldots,\sfX_\ell$ be labelled copies of $\sfX$, with
$\sfX_k=\sfR_k\otimes\sfC_k$, and let
$\sfH_\ell=\sfA\otimes\sfX_1\otimes\cdots\otimes\sfX_\ell$.  Let $Q_k$ be
the operator $Q$ acting on $\sfA$ and $\sfX_k$, tensored with the identity
on the other slots; all $Q_k$ act on the same external factor $\sfA$.  We
claim that
\begin{equation}
  0\preceq\sum_{k=1}^\ell Q_k\preceq(d+\ell-1)\,I_{\sfH_\ell}.
  \label{eq:block-gram}
\end{equation}
To prove this, let $\sfR_k^0$ be a further copy of $\sfR$, with basis
vectors $e_r^{0,k}$, and set
\begin{equation}
  \sfW_k=\sfR_k^0\otimes\bigotimes_{j\ne k}\sfX_j ,
  \qquad
  V_k\Bigl(e_r^{0,k}\otimes\bigotimes_{j\ne k}x_j\Bigr)
  =\sum_{i=1}^nu_i\otimes x_1\otimes\cdots\otimes(e_r\otimes e_i)_k
    \otimes\cdots\otimes x_\ell ,
  \label{eq:Vk-definition}
\end{equation}
which defines a linear map $V_k:\sfW_k\to\sfH_\ell$; the factors are
placed in slot order.  Thus $V_k$ inserts $w_r$ into the
external factor and slot $k$ and keeps the other slots.  Its adjoint is
\begin{equation}
  V_k^\dagger\Bigl(x\otimes\bigotimes_{j=1}^\ell x_j\Bigr)
  =\sum_{r=1}^m\sum_{i=1}^n\langle u_i,x\rangle\,
   \langle e_r\otimes e_i,x_k\rangle\,
   e_r^{0,k}\otimes\bigotimes_{j\ne k}x_j .
  \label{eq:Vk-adjoint}
\end{equation}
Composing \eqref{eq:Vk-definition} with \eqref{eq:Vk-adjoint} and comparing
with \eqref{eq:Q-action} gives $V_kV_k^\dagger=Q_k$, and since the $w_r$ are
orthogonal with squared norm $d$,
\begin{equation}
  V_k^\dagger V_k=d\,I_{\sfW_k}.
  \label{eq:Vk-diagonal}
\end{equation}
For $k\ne l$ we compute $V_k^\dagger V_l$ on a basis vector
$e_{r'}^{0,l}\otimes(e_r\otimes e_j)_k\otimes\varphi$ of $\sfW_l$, where
$\varphi$ collects the slots other than $k$ and $l$.  First $V_l$ produces
$\sum_iu_i\otimes(e_r\otimes e_j)_k\otimes(e_{r'}\otimes e_i)_l\otimes\varphi$;
then $V_k^\dagger$ contracts the external factor against $u_j$ and reads
the row $r$ from slot $k$, so that
\begin{equation}
  V_k^\dagger V_l\bigl(e_{r'}^{0,l}\otimes(e_r\otimes e_j)_k\otimes\varphi\bigr)
  =e_r^{0,k}\otimes\bigl(e_{r'}\otimes\Gamma e_j\bigr)_l\otimes\varphi,
  \qquad \Gamma=UU^\T .
  \label{eq:Vk-cross-factor}
\end{equation}
Thus $V_k^\dagger V_l$ exchanges the two row registers and applies the
orthogonal projection $\Gamma$ to the column register of slot $l$; in
particular $\|V_k^\dagger V_l\|\le1$.  Now let
$V=[V_1\ \cdots\ V_\ell]:\bigoplus_k\sfW_k\to\sfH_\ell$.  Then
$VV^\dagger=\sum_kV_kV_k^\dagger=\sum_kQ_k$, while $V^\dagger V$ is the
block matrix with entries $V_k^\dagger V_l$, so that for
$z=(z_1,\ldots,z_\ell)$,
\[
  \langle z,V^\dagger Vz\rangle
  \le d\sum_k\|z_k\|^2+\sum_{k\ne l}\|z_k\|\,\|z_l\|
  \le(d+\ell-1)\sum_k\|z_k\|^2 ,
\]
using $\sum_{k\ne l}\|z_k\|\|z_l\|=(\sum_k\|z_k\|)^2-\sum_k\|z_k\|^2
\le(\ell-1)\sum_k\|z_k\|^2$.  Since $VV^\dagger$ and $V^\dagger V$ have the
same nonzero eigenvalues, \eqref{eq:block-gram} follows.

\medskip\noindent\emph{Step 2: compression to the hard-core block.}
For a site $t=(r,i)$ write $e_t=e_r\otimes e_i\in\sfX$.  For
$S=\{t_1,\ldots,t_\ell\}\subset F_T$ define
\begin{equation}
  \iota_T|S;T\rangle
  =\frac1{\sqrt{\ell!}}\sum_{\pi\in\mathfrak S_\ell}
   e_{t_{\pi(1)}}\otimes\cdots\otimes e_{t_{\pi(\ell)}}
  \in\sfX_1\otimes\cdots\otimes\sfX_\ell .
  \label{eq:hard-core-isometry}
\end{equation}
The $\ell!$ tensor words are distinct and orthonormal, so $\iota_T$ is an
isometry from $\sfK_{T,\ell}$ into the symmetric tensors supported on
distinct sites of $F_T$; put $\Ical_T=I_d\otimes\iota_T$.  Expanding
\eqref{eq:Q-action}, $Q_k=\sum_r\sum_{i,j}U_{ij}\otimes
(|e_{(r,i)}\rangle\langle e_{(r,j)}|)_k\otimes I$.  Fix $x\in\R^d$, a site
$t=(r,j)\in S$, and a target $t'=(r,i)$ in the same row.  Among the $\ell!$
words in \eqref{eq:hard-core-isometry} and the $\ell$ slots there are
exactly $\ell!$ pairs $(\pi,k)$ for which slot $k$ carries $t$, and
replacing $t$ by $t'$ in these words produces each word of the set
$S'=(S\setminus\{t\})\cup\{t'\}$ exactly once.  If
$t'\in T\cup(S\setminus\{t\})$, the resulting words are annihilated by
$\iota_T^\dagger$, since they either involve a site of $T$ or repeat a
site.  Otherwise the two normalizations $(\ell!)^{-1/2}$ combine with the
count $\ell!$ to give coefficient one.  Hence
\begin{equation}
  \Ical_T^\dagger\Bigl(\sum_{k=1}^\ell Q_k\Bigr)\Ical_T
  \bigl(x\otimes|S;T\rangle\bigr)
  =\sum_{(r,j)\in S}\
   \sum_{\substack{i\in[n]\\ (r,i)\notin T\cup(S\setminus\{(r,j)\})}}
   U_{ij}x\otimes\bigl|(S\setminus\{(r,j)\})\cup\{(r,i)\};T\bigr\rangle ,
  \label{eq:hard-core-compression-action}
\end{equation}
where the inner sum includes $i=j$.  The right-hand side is exactly
$\Ghat(x\otimes|S;T\rangle)$: $P_{rj}$ removes the light site $(r,j)$, and
$P^\dagger_{ri}$ recreates it at a site $(r,i)$ that is free after the
removal.  Therefore
\begin{equation}
  \Ghat=\Ical_T^\dagger\Bigl(\sum_{k=1}^\ell Q_k\Bigr)\Ical_T
  \qquad\text{on }\R^d\otimes\sfK_{T,\ell}.
  \label{eq:hard-core-compression}
\end{equation}

\medskip\noindent\emph{Step 3: conclusion.}
Compressing \eqref{eq:block-gram} by the isometry $\Ical_T$ gives
\eqref{eq:LSL-block}.  For \eqref{eq:LSL-grade}, decompose $\Kcal_\kappa$
into the invariant blocks $\R^d\otimes\sfK_{T,\ell}$ with
$\ell=\kappa-2|T|\le\kappa$: on blocks with $\ell=0$ the operator
vanishes, and on the others $d+\ell-1\le d+\kappa-1$.  Since $d\ge1$, the
constant $d+\kappa-1$ is nonnegative in all cases.
\end{proof}

\begin{remark}[Relation to quantum cloning]
The labelled bound~\eqref{eq:block-gram} is equivalent to a known sharp
spectral bound.  Indeed, with $\Omega_d=\sum_{a=1}^d e_a\otimes e_a$,
\[
  \omega=\sum_i u_i\otimes e_i=(I_d\otimes U)\Omega_d .
\]
After identifying each column-space factor $\operatorname{im}(U)$ through
$U$ and omitting row-register identity factors, $\sum_{k=1}^{\ell}Q_k$
becomes $d\sum_{k=1}^{\ell}P_{0k}$, where $P_{0k}$ projects onto
$\Omega_d/\sqrt d$ on the common external factor and slot $k$, and acts
as the identity on the other slots.  The equal-weight specialization of
Kay, Ramanathan, and
Kaszlikowski~\cite[Section~IV and Appendix~A]{KayRamanathanKaszlikowski2013}
gives $\|d\sum_{k=1}^{\ell}P_{0k}\|=d+\ell-1$ for $\ell\ge1$.
Thus the labelled bound is attained and cannot be improved before
compression.  For $n>d$, decomposing each column slot into
$\operatorname{im}(U)$ and its orthogonal complement leaves the same sum
with $j\le\ell$ active slots, bounded by $d+j-1$ when $j\ge1$ and zero
when $j=0$.  The proof above gives an elementary insertion-map argument
and then compresses to the distinct-site sector; compression preserves
the upper bound, but need not preserve equality.
\end{remark}

\section{Row-local hop estimates}
\label{sec:hop}

This section bounds operators that act at one or two sites of a single row,
with the external action determined by the frame vectors of that row.  The
estimates depend only on counting sites at a given level.  For
$\psi\in\R^d\otimes\Fcal$ write $\psi=\sum_\tau x_\tau\otimes|\tau\rangle$
with $x_\tau\in\R^d$, so that $\|\psi\|^2=\sum_\tau\|x_\tau\|^2$.  For an
operator $\Tcal$ with values in $\R^d\otimes\Fcal$ and a pattern $\sigma$,
let $(\Tcal\psi)_\sigma=(I_d\otimes\langle\sigma|)\Tcal\psi\in\R^d$ denote
the coefficient of $|\sigma\rangle$; for an operator with values in
$\Fcal$, $(\Tcal\psi)_\sigma=\langle\sigma|\Tcal\psi\rangle\in\R$.  We use
the row-subset inequalities \eqref{eq:subset-parseval} and the elementary
inequality $\|\sum_{\alpha\in I}z_\alpha\|^2\le|I|\sum_{\alpha\in I}\|z_\alpha\|^2$.
For a row $r$ let
\[
  \nu_r(\tau)=|S_1^r(\tau)|+2|S_2^r(\tau)|
\]
be the grade of row $r$ in the pattern $\tau$, so that
$\sum_r\nu_r(\tau)=\gr(\tau)$, and let $\Ncal_r$ be the operator on
$\R^d\otimes\Fcal$ with $\Ncal_r(x\otimes|\tau\rangle)=\nu_r(\tau)\,x\otimes|\tau\rangle$.

In the following lemma, $Z=|a'\rangle\langle a|$ and $W=|c'\rangle\langle c|$
denote arbitrary legs, so that $Z$ changes the level of a site from $a$ to
$a'$ and $W$ from $c$ to $c'$; we use the convention $\max\varnothing=0$.

\begin{lemma}[Row-local hop estimate]
\label{lem:hop}
Let $Z=|a'\rangle\langle a|$ and $W=|c'\rangle\langle c|$ be legs.
\begin{enumerate}[label=\textup{(\roman*)}]
\item \textup{(One site.)}  Fix a row $r$ and an integer $\ell\ge0$, and let
  $\psi$ be supported on patterns $\tau$ with $\nu_r(\tau)=\ell$.  Then the
  one-site map $\Tcal_r=\sum_iu_i^\T\otimes Z_{ri}=\Scal_r(Z^\dagger)^\dagger$
  satisfies
  \begin{equation}
    \|\Tcal_r\psi\|^2\le N\,\|\psi\|^2,
    \qquad
    N=\max\bigl\{|S_{a'}^r(\sigma)|:\nu_r(\sigma)=\ell+a'-a\bigr\}.
    \label{eq:hop-one-site}
  \end{equation}
\item \textup{(Two sites.)}  Let
  $\Tcal=\sum_r\sum_{i\ne j}U_{ij}\otimes Z_{ri}W_{rj}$.  For every
  $\nu\ge0$ and $\psi\in\Kcal_\nu$,
  \begin{equation}
    \|\Tcal\psi\|^2\le N_{\mathrm{out}}\,N_{\mathrm{in}}\,\|\psi\|^2,
    \label{eq:hop-two-site}
  \end{equation}
  where $N_{\mathrm{in}}=\max\{|S_a(\tau)|:\gr(\tau)=\nu\}$ and
  $N_{\mathrm{out}}=\max\{|S_{c'}(\sigma)|:\gr(\sigma)=\nu+(a'-a)+(c'-c),\
  S_{a'}(\sigma)\ne\varnothing\}$.
\end{enumerate}
\end{lemma}

\begin{proof}
(i) For a pattern $\sigma$ with $\nu_r(\sigma)=\ell+a'-a$, the coefficient
of $|\sigma\rangle$ in $\Tcal_r\psi$ collects the inputs that differ from
$\sigma$ at one site of row $r$ where $\sigma$ has level $a'$ and the input
has level $a$:
\[
  (\Tcal_r\psi)_\sigma=\sum_{i\in S_{a'}^r(\sigma)}u_i^\T x_{\sigma[i\mapsto a]},
\]
with $x_\tau=0$ when $\tau$ is outside the support of $\psi$.  Hence
\begin{align*}
  \|\Tcal_r\psi\|^2
  &=\sum_\sigma\Bigl|\sum_{i\in S_{a'}^r(\sigma)}u_i^\T x_{\sigma[i\mapsto a]}\Bigr|^2
  \le N\sum_\sigma\sum_{i\in S_{a'}^r(\sigma)}
     \bigl(u_i^\T x_{\sigma[i\mapsto a]}\bigr)^2\\
  &=N\sum_\tau\sum_{i\in S_a^r(\tau)}(u_i^\T x_\tau)^2
  \le N\sum_\tau\|x_\tau\|^2 ,
\end{align*}
where the equality is the change of variables
$(\sigma,i)\mapsto(\tau,i)=(\sigma[i\mapsto a],i)$, a bijection onto pairs
with $i\in S_a^r(\tau)$ whose inverse is $\sigma=\tau[i\mapsto a']$, and the
last step is the analysis inequality in \eqref{eq:subset-parseval}.

(ii) For an output pattern $\sigma$, the coefficient $(\Tcal\psi)_\sigma$
collects the inputs $\tau$ and ordered pairs $(i,j)$ with $\tau_{ri}=a$,
$\tau_{rj}=c$, and $\tau[i\mapsto a',j\mapsto c']=\sigma$; given
$\sigma,r,i,j$, the input is $\tau=\sigma[i\mapsto a,j\mapsto c]$.  Hence
\[
  (\Tcal\psi)_\sigma=\sum_r\sum_{j\in S_{c'}^r(\sigma)}z_{\sigma,r,j},
  \qquad
  z_{\sigma,r,j}=\sum_{i\in S_{a'}^r(\sigma)\setminus\{j\}}
    u_i\,\bigl(u_j^\T x_{\sigma[i\mapsto a,\,j\mapsto c]}\bigr).
\]
If $S_{a'}(\sigma)=\varnothing$ the coefficient vanishes, so we may restrict
to output patterns with $S_{a'}(\sigma)\ne\varnothing$, for which the number
of pairs $(r,j)$ in the outer sum is $|S_{c'}(\sigma)|\le N_{\mathrm{out}}$.
By the finite-family inequality, and then the synthesis inequality in
\eqref{eq:subset-parseval} applied for fixed $(\sigma,r,j)$ to the distinct
columns $i$ of row $r$,
\[
  \|\Tcal\psi\|^2
  \le N_{\mathrm{out}}\sum_\sigma\sum_r\sum_{j\in S_{c'}^r(\sigma)}\|z_{\sigma,r,j}\|^2
  \le N_{\mathrm{out}}\sum_\sigma\sum_r\sum_{j\in S_{c'}^r(\sigma)}
     \sum_{i\in S_{a'}^r(\sigma)\setminus\{j\}}
     \bigl(u_j^\T x_{\sigma[i\mapsto a,\,j\mapsto c]}\bigr)^2 .
\]
The map $(\sigma,r,i,j)\mapsto(\tau,r,i,j)$ with
$\tau=\sigma[i\mapsto a,j\mapsto c]$ is a bijection from the index set of
the last sum onto the set of $(\tau,r,i,j)$ with $\gr(\tau)=\nu$,
$i\in S_a^r(\tau)$, $j\in S_c^r(\tau)$, and $i\ne j$; its inverse is
$\sigma=\tau[i\mapsto a',j\mapsto c']$.  Re-indexing, and applying the
analysis inequality in \eqref{eq:subset-parseval} to the sum over $j$ for
fixed $(\tau,r,i)$,
\[
  \|\Tcal\psi\|^2
  \le N_{\mathrm{out}}\sum_\tau\sum_r\sum_{i\in S_a^r(\tau)}
     \sum_{j\in S_c^r(\tau)\setminus\{i\}}(u_j^\T x_\tau)^2
  \le N_{\mathrm{out}}\sum_\tau|S_a(\tau)|\,\|x_\tau\|^2
  \le N_{\mathrm{out}}N_{\mathrm{in}}\|\psi\|^2 .
  \qedhere
\]
\end{proof}

\begin{corollary}[Directional bounds]
\label{cor:directional}
For every $\nu\ge0$,
\begin{equation}
  \|Y_+|_{\Kcal_\nu}\|\le\nu,\qquad
  \|Y_0|_{\Kcal_\nu}\|\le\frac{\nu}{\sqrt2},\qquad
  \|H'|_{\Kcal_\nu}\|\le\nu .
  \label{eq:directional-bounds}
\end{equation}
\end{corollary}

\begin{proof}
Apply Lemma~\ref{lem:hop}(ii).  For $Y_+$ we have $(Z,W)=(R,P^\dagger)$,
so $(a,a',c,c')=(1,2,0,1)$: the outputs have grade $\nu+2$ and at least one
heavy site, hence $|S_1(\sigma)|=\nu+2-2|S_2(\sigma)|\le\nu$, while
$|S_1(\tau)|\le\nu$ for inputs; thus $\|Y_+\psi\|^2\le\nu^2\|\psi\|^2$.  For
$Y_0$ we have $(Z,W)=(R^\dagger,P^\dagger)$ and $(a,a',c,c')=(2,1,0,1)$: the
outputs have grade $\nu$, so $|S_1(\sigma)|\le\nu$, while
$|S_2(\tau)|\le\nu/2$ for inputs; thus $\|Y_0\psi\|^2\le\nu^2\|\psi\|^2/2$.
For $H'$ we have $(Z,W)=(P,P^\dagger)$ and $(a,a',c,c')=(1,0,0,1)$: outputs
have grade $\nu$ and $|S_1(\sigma)|\le\nu$, and $|S_1(\tau)|\le\nu$ for
inputs.
\end{proof}

In each case the counts that enter are the number of light sites of the
output and the number of level-$a$ sites of the input, both of order $\nu$.
Applied to $L_+$, the same lemma would produce the number of free sites of
the input, which is of order $mn$; for this reason $L_+$ is treated by
Lemma~\ref{lem:LSL} instead.  Table~\ref{tab:hop} summarizes the five
applications of Lemma~\ref{lem:hop} used in this paper.

\begin{table}[ht]
\centering
\small
\renewcommand{\arraystretch}{1.25}
\begin{tabular}{lcccll}
\toprule
Operator & $(Z,W)$ or $Z$ & $(a,a')$ & $(c,c')$ & Counts & Bound\\
\midrule
$Y_+$ & $(R,P^\dagger)$ & $(1,2)$ & $(0,1)$
  & $N_{\mathrm{out}}\le\nu$, $N_{\mathrm{in}}\le\nu$ & $\|Y_+|_{\Kcal_\nu}\|\le\nu$\\
$Y_0$ & $(R^\dagger,P^\dagger)$ & $(2,1)$ & $(0,1)$
  & $N_{\mathrm{out}}\le\nu$, $N_{\mathrm{in}}\le\nu/2$ & $\|Y_0|_{\Kcal_\nu}\|\le\nu/\sqrt2$\\
$H'$ & $(P,P^\dagger)$ & $(1,0)$ & $(0,1)$
  & $N_{\mathrm{out}}\le\nu$, $N_{\mathrm{in}}\le\nu$ & $\|H'|_{\Kcal_\nu}\|\le\nu$\\
\midrule
$\Scal_r(R)^\dagger$ & $R^\dagger$ & $(2,1)$ & ---
  & $N\le(\ell-1)_+$ & $\Scal_r(R)\Scal_r(R)^\dagger\preceq\Ncal_r$\\
$\Scal_r(R^\dagger)^\dagger$ & $R$ & $(1,2)$ & ---
  & $N\le\lfloor(\ell+1)/2\rfloor$ & $\Scal_r(R^\dagger)\Scal_r(R^\dagger)^\dagger\preceq\Ncal_r$\\
\bottomrule
\end{tabular}
\caption{The applications of Lemma~\ref{lem:hop}.  The first three rows use
part~(ii) on $\Kcal_\nu$ (Corollary~\ref{cor:directional}); the outputs
have grade $\nu+2$, $\nu$, and $\nu$, respectively.  The last two rows use
part~(i) on the row-grade stratum $\ell$ (Corollary~\ref{cor:heavy-gram}),
where the image lies in the stratum $\ell-1$, respectively $\ell+1$; in
both cases $N\le\ell$ for $\ell\ge1$, and the operator vanishes on the
stratum $\ell=0$.}
\label{tab:hop}
\end{table}

\begin{corollary}[Heavy Gram budgets]
\label{cor:heavy-gram}
For every row $r$,
\[
  \Scal_r(R)\Scal_r(R)^\dagger\preceq\Ncal_r
  \qquad\text{and}\qquad
  \Scal_r(R^\dagger)\Scal_r(R^\dagger)^\dagger\preceq\Ncal_r .
\]
Consequently,
\begin{equation}
  \Gcal(R)\preceq\nu\,I\quad\text{and}\quad\Gcal(R^\dagger)\preceq\nu\,I
  \qquad\text{on }\Kcal_\nu .
  \label{eq:heavy-gram}
\end{equation}
\end{corollary}

\begin{proof}
Write $\psi=\sum_\ell\psi_{r,\ell}$ with $\psi_{r,\ell}$ supported on
patterns with $\nu_r(\tau)=\ell$.  The map
$\Scal_r(R)^\dagger=\sum_iu_i^\T\otimes R^\dagger_{ri}$ is the one-site map
of Lemma~\ref{lem:hop}(i) with $Z=R^\dagger$, $(a,a')=(2,1)$; it sends the
stratum $\ell$ into the stratum $\ell-1$, so the images of different strata
are orthogonal, and
\begin{align*}
  \langle\psi,\Scal_r(R)\Scal_r(R)^\dagger\psi\rangle
  &=\sum_{\ell\ge1}\|\Scal_r(R)^\dagger\psi_{r,\ell}\|^2
  \le\sum_{\ell\ge1}\max\{|S_1^r(\sigma)|:\nu_r(\sigma)=\ell-1\}\,
       \|\psi_{r,\ell}\|^2\\
  &\le\sum_{\ell\ge1}\ell\,\|\psi_{r,\ell}\|^2
  =\langle\psi,\Ncal_r\psi\rangle ,
\end{align*}
where the $\ell=0$ term vanishes and, for $\ell\ge1$,
$|S_1^r(\sigma)|\le\nu_r(\sigma)=\ell-1\le\ell$.  Similarly,
$\Scal_r(R^\dagger)^\dagger=\sum_iu_i^\T\otimes R_{ri}$ has $(a,a')=(1,2)$
and sends the stratum $\ell$ into the stratum $\ell+1$; the relevant count
is $|S_2^r(\sigma)|\le\lfloor(\ell+1)/2\rfloor\le\ell$ for $\ell\ge1$, while
for $\ell=0$ the row has no light site and the image is zero.  Summing over
$r$ and using $\sum_r\Ncal_r=\nu I$ on $\Kcal_\nu$ gives
\eqref{eq:heavy-gram}.
\end{proof}

\section{Band estimates}
\label{sec:band-bounds}

\subsection{A block Cauchy--Schwarz inequality}

\begin{lemma}
\label{lem:block-CS}
Let $\mathcal X,\mathcal Y,\mathcal Z$ be finite-dimensional Hilbert
spaces, let $\Pcal_{\mathrm{in}},\Pcal_{\mathrm{mid}},\Pcal_{\mathrm{out}}$
be orthogonal projections on them, and let $A_r:\mathcal X\to\mathcal Y$
and $B_r:\mathcal Y\to\mathcal Z$ for $r\in[m]$.  Then
\begin{equation}
  \Bigl\|\Pcal_{\mathrm{out}}\Bigl(\sum_rB_r\Pcal_{\mathrm{mid}}A_r\Bigr)
    \Pcal_{\mathrm{in}}\Bigr\|^2
  \le\Bigl\|\Pcal_{\mathrm{in}}\Bigl(\sum_rA_r^\dagger\Pcal_{\mathrm{mid}}A_r\Bigr)
    \Pcal_{\mathrm{in}}\Bigr\|
  \cdot\Bigl\|\Pcal_{\mathrm{out}}\Bigl(\sum_rB_r\Pcal_{\mathrm{mid}}B_r^\dagger\Bigr)
    \Pcal_{\mathrm{out}}\Bigr\| .
  \label{eq:block-CS}
\end{equation}
\end{lemma}

\begin{proof}
Let $\mathcal A:\operatorname{ran}\Pcal_{\mathrm{in}}\to
\bigoplus_r\operatorname{ran}\Pcal_{\mathrm{mid}}$,
$\mathcal Ax=(\Pcal_{\mathrm{mid}}A_rx)_r$, and
$\mathcal B:\bigoplus_r\operatorname{ran}\Pcal_{\mathrm{mid}}\to
\operatorname{ran}\Pcal_{\mathrm{out}}$,
$\mathcal B(y_r)_r=\sum_r\Pcal_{\mathrm{out}}B_ry_r$.  Then
$\mathcal B\mathcal A$ is the operator on the left of \eqref{eq:block-CS},
while $\mathcal A^\dagger\mathcal A$ and $\mathcal B\mathcal B^\dagger$ are
the two operators on the right, and
$\|\mathcal B\mathcal A\|^2\le\|\mathcal A\|^2\|\mathcal B\|^2
=\|\mathcal A^\dagger\mathcal A\|\,\|\mathcal B\mathcal B^\dagger\|$.
\end{proof}

We apply the lemma to the factorized families of
Corollary~\ref{cor:factorization}.  For a leg $Z$ let
$\delta_Z\in\{1,-1\}$ be its grade shift, so that both $\Scal_r(Z)$ and
$\Scal_r(Z^\dagger)^\dagger=\sum_iu_i^\T\otimes Z_{ri}$ shift the grade by
$\delta_Z$.

\begin{corollary}
\label{cor:leg-CS}
For legs $Z,W$ and every $\nu\ge0$,
\begin{equation}
  \Bigl\|\Bigl(\sum_r\Scal_r(Z)\Scal_r(W^\dagger)^\dagger\Bigr)\Big|_{\Kcal_\nu}\Bigr\|^2
  \le\bigl\|\Gcal(W^\dagger)|_{\Kcal_\nu}\bigr\|
    \cdot\bigl\|\Gcal(Z)|_{\Kcal_{\nu+\delta_W+\delta_Z}}\bigr\| .
  \label{eq:leg-CS}
\end{equation}
\end{corollary}

\begin{proof}
Apply Lemma~\ref{lem:block-CS} with $A_r=\Scal_r(W^\dagger)^\dagger$,
$B_r=\Scal_r(Z)$, $\Pcal_{\mathrm{in}}=\Pcal_\nu$,
$\Pcal_{\mathrm{mid}}=\pi_{\nu+\delta_W}$, and
$\Pcal_{\mathrm{out}}=\Pcal_{\nu+\delta_W+\delta_Z}$.  By grade
homogeneity, $A_r\Pcal_\nu=\pi_{\nu+\delta_W}A_r\Pcal_\nu$ and
$B_r\pi_{\nu+\delta_W}=\Pcal_{\nu+\delta_W+\delta_Z}B_r\pi_{\nu+\delta_W}$,
so the left side of \eqref{eq:block-CS} is the left side of
\eqref{eq:leg-CS}.  On the right, since $\pi_{\nu+\delta_W}\preceq I$,
\[
  \sum_rA_r^\dagger\pi_{\nu+\delta_W}A_r\preceq\sum_rA_r^\dagger A_r=\Gcal(W^\dagger),
  \qquad
  \sum_rB_r\pi_{\nu+\delta_W}B_r^\dagger\preceq\sum_rB_rB_r^\dagger=\Gcal(Z);
\]
both Gram operators preserve grade, and compressing a Loewner inequality
between positive operators to a grade sector preserves it.
\end{proof}

\subsection{The four Gram budgets}

\begin{proposition}
\label{prop:gram-budgets}
For every $\kappa\ge0$, on $\Kcal_\kappa$,
\begin{equation}
  \Gcal(P^\dagger)=\Ghat\preceq(d+\kappa-1)\,I,\qquad
  \Gcal(P)\preceq(m+\kappa)\,I,\qquad
  \Gcal(R)\preceq\kappa\,I,\qquad
  \Gcal(R^\dagger)\preceq\kappa\,I .
  \label{eq:gram-budgets}
\end{equation}
Moreover $0\preceq D_0\preceq mI$ on $\R^d\otimes\Fcal$, while
$0\preceq D_1\preceq\kappa I$ and $0\preceq D_H\preceq\kappa I$ on
$\Kcal_\kappa$.
\end{proposition}

\begin{proof}
The first bound is Lemma~\ref{lem:LSL}, and the last two are
Corollary~\ref{cor:heavy-gram}.  The diagonal operators of
\eqref{eq:D-ops} are block diagonal in the pattern basis: on
$x\otimes|\tau\rangle$ they act on the external vector by
$\sum_r\sum_{i\in F_r(\tau)}u_iu_i^\T$, by $\sum_{(r,i)\in S_1(\tau)}u_iu_i^\T$,
and by $\sum_{(r,i)\in S_1(\tau)\cup S_2(\tau)}u_iu_i^\T$, respectively.
Each row contributes a matrix $\sum_{i\in S}u_iu_i^\T\preceq I_d$ by
\eqref{eq:subset-parseval}; there are $m$ rows in the first case, while in
the other two cases the number of contributing rows is at most the number
of occupied sites, which is at most $\gr(\tau)=\kappa$.  Finally,
$\Gcal(P)=D_0+H'$ by \eqref{eq:fact-light}, and since $H'$ is self-adjoint
with $\|H'|_{\Kcal_\kappa}\|\le\kappa$ by Corollary~\ref{cor:directional},
$H'\preceq\kappa I$ on $\Kcal_\kappa$.
\end{proof}

\subsection{Heavy--heavy terms}

By \eqref{eq:fact-plus}, $H_+=\sum_r\Scal_r(R)\Scal_r(R^\dagger)^\dagger$.
Corollary~\ref{cor:leg-CS} with $Z=W=R$ and
Proposition~\ref{prop:gram-budgets} give
\begin{equation}
  \|H_+|_{\Kcal_\nu}\|^2
  \le\|\Gcal(R^\dagger)|_{\Kcal_\nu}\|\cdot\|\Gcal(R)|_{\Kcal_{\nu+2}}\|
  \le\nu(\nu+2),
  \qquad\text{so}\qquad
  \|H_+|_{\Kcal_\nu}\|\le\nu+1 .
  \label{eq:Hplus}
\end{equation}
By \eqref{eq:fact-heavy} and Proposition~\ref{prop:gram-budgets}, on
$\Kcal_\nu$,
\begin{equation}
  -\nu I\preceq-D_H\preceq H_0=\Gcal(R)+\Gcal(R^\dagger)-D_H\preceq2\nu I,
  \qquad\text{so}\qquad
  \|H_0|_{\Kcal_\nu}\|\le2\nu .
  \label{eq:H0}
\end{equation}
Since $h/m=(b-1)/(sb)\le1/s$,
\begin{equation}
  \Bigl\|\frac hmH_+\Big|_{\Kcal_\nu}\Bigr\|\le\frac{\nu+1}{s},
  \qquad
  \Bigl\|\frac hmH_0\Big|_{\Kcal_\nu}\Bigr\|\le\frac{2\nu}{s}.
  \label{eq:heavy-normalized}
\end{equation}

\subsection{Light--light terms}

By \eqref{eq:fact-plus}, $L_+=\sum_r\Scal_r(P^\dagger)\Scal_r(P)^\dagger$.
Corollary~\ref{cor:leg-CS} with $Z=W=P^\dagger$ and
Proposition~\ref{prop:gram-budgets} give
\begin{equation}
  \Bigl\|\frac1mL_+\Big|_{\Kcal_\nu}\Bigr\|
  \le\frac1m\sqrt{\|\Gcal(P)|_{\Kcal_\nu}\|\cdot\|\Ghat|_{\Kcal_{\nu+2}}\|}
  \le\frac1m\sqrt{(m+\nu)(d+\nu+1)} .
  \label{eq:Lplus-bound}
\end{equation}
This is the geometric mean of the input budget $m+\nu$ and the output
budget $d+\nu+1$, divided by $m$.  For the grade-preserving term,
\eqref{eq:fact-light} and Proposition~\ref{prop:gram-budgets} give
\begin{equation}
  \Bigl\|\frac1mL_0\Big|_{\Kcal_\nu}\Bigr\|
  \le\frac{\|\Ghat|_{\Kcal_\nu}\|+\|D_1|_{\Kcal_\nu}\|+\|H'|_{\Kcal_\nu}\|}{m}
  \le\frac{(d+\nu-1)+2\nu}{m}
  \le\frac{3(d+\nu)}{m}\le\frac{3(d+\nu+1)}{m}.
  \label{eq:L0-bound}
\end{equation}

\begin{remark}[Sharpness at $\nu=0$]
\label{rem:Lplus-sharp}
The bound \eqref{eq:Lplus-bound} cannot be improved in general.  For
$x\in\R^d$, the vector $L_+(x\otimes\Omega)$ has, for each row $r$ and each
unordered pair $\{i,j\}$, the coefficient $(U_{ij}+U_{ji})x$ on the basis
vector with light sites $(r,i)$ and $(r,j)$, so a direct computation using
\eqref{eq:parseval-frame} gives
\[
  \|L_+(x\otimes\Omega)\|^2
  =m\Bigl[(d+1)\|x\|^2-2\sum_{i=1}^n\|u_i\|^2(u_i^\T x)^2\Bigr].
\]
Hence $\|m^{-1}L_+|_{\Kcal_0}\|^2\le(d+1)/m$, with equality in the limit of
frames with $\max_i\|u_i\|\to0$; this matches \eqref{eq:Lplus-bound} at
$\nu=0$.
\end{remark}

\subsection{Mixed terms}

By \eqref{eq:fact-mixed},
$X_+=\sum_r\Scal_r(P^\dagger)\Scal_r(R^\dagger)^\dagger$ and
$X_0=\sum_r\Scal_r(P^\dagger)\Scal_r(R)^\dagger$.  Corollary~\ref{cor:leg-CS}
with $(Z,W)=(P^\dagger,R)$ and $(P^\dagger,R^\dagger)$, together with
Proposition~\ref{prop:gram-budgets}, gives
\begin{equation}
  \|X_+|_{\Kcal_\nu}\|\le\sqrt{\nu(d+\nu+1)},\qquad
  \|X_0|_{\Kcal_\nu}\|\le\sqrt{\nu(d+\nu-1)}\le\sqrt{\nu(d+\nu)} .
  \label{eq:X-bounds}
\end{equation}
The bounds for $Y_+$ and $Y_0$ are those of Corollary~\ref{cor:directional}.
Since $X_0$ and $Y_0$ map $\Kcal_\nu$ into itself, $X_0^\dagger$ and
$Y_0^\dagger$ have the same norms on $\Kcal_\nu$.  To absorb the
coefficient $\sqrt h$ we use
\begin{equation}
  \frac{\sqrt h}{m}\sqrt{\nu(d+\nu+1)}
  \le\frac12\Bigl(\frac{d+\nu+1}{m}+\frac{h\nu}{m}\Bigr),
  \qquad
  \frac{h\nu}{m}\le\frac{\nu}{s},
  \qquad
  \frac{\sqrt h\,\nu}{m}\le\frac{\nu}{s},
  \label{eq:mixed-absorb}
\end{equation}
which follow from $2\sqrt{xy}\le x+y$, from $h/m\le1/s$, and from
$\sqrt{b-1}\le b=m/s$.  Consequently
\begin{align}
  \Bigl\|\frac{\sqrt h}m(X_++Y_+)\Big|_{\Kcal_\nu}\Bigr\|
  &\le\frac{d+\nu+1}{m}+\frac{3\nu}{2s},
  \label{eq:mixed-plus}\\
  \Bigl\|\frac{\sqrt h}m(X_0+X_0^\dagger+Y_0+Y_0^\dagger)\Big|_{\Kcal_\nu}\Bigr\|
  &\le\frac{d+\nu}{m}+(1+\sqrt2)\frac{\nu}{s} .
  \label{eq:mixed-zero}
\end{align}
(The first term in \eqref{eq:mixed-plus} could be halved; we keep the
stated form for the uniform table below.)

\subsection{The uniform band envelope}

Set
\begin{equation}
  \rho_\nu=\sqrt{\frac{d+\nu+1}{m}},\qquad
  \gamma_\nu=\frac{\nu+1}{s}.
  \label{eq:envelope-scales}
\end{equation}
Since $\sqrt{\nu/m}\le\rho_\nu$, the bound \eqref{eq:Lplus-bound} gives
\begin{equation}
  \Bigl\|\frac1mL_+\Big|_{\Kcal_\nu}\Bigr\|
  \le\rho_\nu\sqrt{1+\frac\nu m}
  \le\rho_\nu\Bigl(1+\sqrt{\frac\nu m}\Bigr)
  \le\rho_\nu+\rho_\nu^2 ,
  \label{eq:Lplus-envelope}
\end{equation}
with no restriction on the relative size of $\nu$ and $m$.  The estimates
of this section are collected in Table~\ref{tab:bands}.

\begin{table}[ht]
\centering
\renewcommand{\arraystretch}{1.3}
\begin{tabular}{ccccl}
\toprule
Grade shift & Light--light & Light--heavy & Heavy--heavy & Source\\
\midrule
$+2$ & $\rho_\nu+\rho_\nu^2$ & $\rho_\nu^2+\tfrac32\gamma_\nu$ & $\gamma_\nu$
  & \eqref{eq:Lplus-envelope}, \eqref{eq:mixed-plus}, \eqref{eq:heavy-normalized}\\
$0$  & $3\rho_\nu^2$ & $\rho_\nu^2+(1+\sqrt2)\gamma_\nu$ & $2\gamma_\nu$
  & \eqref{eq:L0-bound}, \eqref{eq:mixed-zero}, \eqref{eq:heavy-normalized}\\
$-2$ & $\rho_\nu+\rho_\nu^2$ & $\rho_\nu^2+\tfrac32\gamma_\nu$ & $\gamma_\nu$
  & adjoint of the $+2$ row\\
\bottomrule
\end{tabular}
\caption{Upper bounds on $\Kcal_\nu$ for the norms of the three components
of $B^+$, $B^0$, and $B^-$ in \eqref{eq:exact-inventory}, including the
normalization $1/m$ and the coefficients $\sqrt h$ and $h$.  Each entry of
the last row vanishes when $\nu<2$; for $\nu\ge2$ it is the corresponding
$+2$ entry at grade $\nu-2$, because $B^-|_{\Kcal_\nu}$ is the adjoint of
$B^+|_{\Kcal_{\nu-2}}$, and $\rho_\nu$, $\gamma_\nu$ are nondecreasing in
$\nu$.}
\label{tab:bands}
\end{table}

\begin{proposition}[Uniform band envelope]
\label{prop:band-envelope}
For every $\nu\ge0$ and $\Delta\in\{1,0,-1\}$,
\begin{equation}
  \|B^\Delta|_{\Kcal_\nu}\|
  \le\Cband\Bigl[\sqrt{\frac{d+\nu+1}{m}}+\frac{d+\nu+1}{m}+\frac{\nu+1}{s}\Bigr],
  \qquad \Cband=3+\sqrt2 .
  \label{eq:band-envelope}
\end{equation}
\end{proposition}

\begin{proof}
By \eqref{eq:exact-inventory}, the triangle inequality, and
Table~\ref{tab:bands},
$\|B^\pm|_{\Kcal_\nu}\|\le\rho_\nu+2\rho_\nu^2+\tfrac52\gamma_\nu$ and
$\|B^0|_{\Kcal_\nu}\|\le4\rho_\nu^2+(3+\sqrt2)\gamma_\nu$.  Every
coefficient is at most $3+\sqrt2$.
\end{proof}

\section{Proof of the main theorem}
\label{sec:proof}

\subsection{The moment bound}

Fix an integer $q\ge1$, let $D_q=d+2q+1$, and define
\begin{equation}
  \beta_q=\Cband\Bigl[\sqrt{\frac{D_q}{m}}+\frac{D_q}{m}+\frac{2q+1}{s}\Bigr].
  \label{eq:beta}
\end{equation}
By Proposition~\ref{prop:band-envelope},
$\|B^\Delta|_{\Kcal_\nu}\|\le\beta_q$ for all $\Delta$ and all
$0\le\nu\le2q$.

\begin{theorem}[Moment bound]
\label{thm:moment}
For every integer $q\ge1$,
\begin{equation}
  \E\tr(E_\iid^{2q})\le d\,(3\beta_q)^{2q}
  \qquad\text{and}\qquad
  \E\tr(E_\SStack^{2q})\le d\,(3c_b^2\beta_q)^{2q}.
  \label{eq:moment-bounds}
\end{equation}
\end{theorem}

\begin{proof}
Let $v_\alpha=q_\alpha\otimes\Omega\in\Kcal_0$.  Inserting
$I=\sum_\nu\Pcal_\nu$ between consecutive factors of $\Mcal^{2q}$ in
\eqref{eq:vacuum-moment} and using $\Mcal=B^++B^0+B^-$ with
$B^\Delta\Kcal_\nu\subseteq\Kcal_{\nu+2\Delta}$,
\[
  \langle v_\alpha,\Mcal^{2q}v_\alpha\rangle
  =\sum_{(\Delta_1,\ldots,\Delta_{2q})\in\{1,0,-1\}^{2q}}
   \langle v_\alpha,B^{\Delta_{2q}}\cdots B^{\Delta_1}v_\alpha\rangle ,
\]
where the word indexed by $(\Delta_t)$ visits the grades
$\nu_t=2\sum_{u\le t}\Delta_u$ and vanishes unless $\nu_{2q}=0$ and
$\nu_t\ge0$ for all $t$.  A nonvanishing word has the same number $k\le q$
of upward and downward steps, so $0\le\nu_t\le2k\le2q$ at all times.  For
such a word, submultiplicativity gives
\[
  |\langle v_\alpha,B^{\Delta_{2q}}\cdots B^{\Delta_1}v_\alpha\rangle|
  \le\prod_{t=1}^{2q}\|B^{\Delta_t}|_{\Kcal_{\nu_{t-1}}}\|\le\beta_q^{2q}.
\]
There are $3^{2q}$ words; summing over them and over $\alpha\in[d]$ gives
the first bound in \eqref{eq:moment-bounds}, and
Proposition~\ref{prop:transfer} gives the second.
\end{proof}

\subsection{From moments to probability}

For a real symmetric matrix $A$, $\|A\|^{2q}\le\tr(A^{2q})$, because the
right side is the sum of the $2q$-th powers of the eigenvalues.  Markov's
inequality and Theorem~\ref{thm:moment} therefore give
\begin{equation}
  \Prob\{\|E_\SStack\|>\varepsilon\}
  \le\frac{\E\tr(E_\SStack^{2q})}{\varepsilon^{2q}}
  \le d\Bigl(\frac{3c_b^2\beta_q}{\varepsilon}\Bigr)^{2q}.
  \label{eq:Markov}
\end{equation}
Now let $q$, $s$, $b$, and $m$ be as in \eqref{eq:qD} and
\eqref{eq:integer-parameters}.  Then
\begin{equation}
  m\ge\frac{\Lambda^2D_q}{\varepsilon^2},\qquad
  s\ge\frac{\Lambda(2q+1)}{\varepsilon},
  \label{eq:scale-conditions}
\end{equation}
so that $\sqrt{D_q/m}\le\varepsilon/\Lambda$,
$(2q+1)/s\le\varepsilon/\Lambda$, and
$D_q/m\le\varepsilon^2/\Lambda^2\le\varepsilon/\Lambda$, using
$\varepsilon\le1\le\Lambda$.  Hence
\begin{equation}
  \beta_q\le\frac{3\Cband\,\varepsilon}{\Lambda}
  =\frac{\Cband\,\varepsilon}{30c_\star^2},
  \qquad
  \frac{3c_b^2\beta_q}{\varepsilon}\le\frac{\Cband}{10}
  =\frac{3+\sqrt2}{10}<0.442<\frac12 ,
  \label{eq:beta-small}
\end{equation}
where the second inequality uses $c_b\le c_\star$ from
Lemma~\ref{lem:coupling}(iii).  By \eqref{eq:Markov},
\begin{equation}
  \Prob\{\|E_\SStack\|>\varepsilon\}\le d\,4^{-q}\le\frac{\delta^2}{d}\le\delta ,
  \label{eq:failure-half}
\end{equation}
since $2^{-q}\le\delta/d$ by the definition of $q$, and $\delta\le1\le d$.
This proves \eqref{eq:main-prob}; \eqref{eq:ose} is the equivalent
quadratic-form statement noted in Section~\ref{sec:intro}.

\subsection{Integer rounding and explicit constants}

Put $\lambda=\Lambda(2q+1)/\varepsilon$; then $\lambda\ge1$ and
$s=\lceil\lambda\rceil\le\lambda+1\le2\lambda$.  Since $2q+1\le3q$,
\begin{equation}
  s\le\frac{\Lambda(2q+1)}{\varepsilon}+1\le\frac{6\Lambda q}{\varepsilon}
  <\frac{1356\,q}{\varepsilon}.
  \label{eq:s-rounding}
\end{equation}
Since $b=\lceil M_0/s\rceil$, we have $M_0\le m=sb<M_0+s$, and
\begin{equation}
  \frac{s}{M_0}\le\frac{2\Lambda(2q+1)/\varepsilon}{\Lambda^2D_q/\varepsilon^2}
  \le\frac{2\varepsilon}{\Lambda}<1 ,
  \label{eq:rounding-ratio}
\end{equation}
using $2q+1\le D_q$.  In particular $M_0/s>1$, so
\begin{equation}
  b=\Bigl\lceil\frac{M_0}{s}\Bigr\rceil\ge2 ,
  \label{eq:b-at-least-two}
\end{equation}
which is the fact used in Section~\ref{sec:fock-model}, and
$m<M_0+s<2M_0$.  Using $2M_0<3\Lambda^2D_q/\varepsilon^2$ and
$D_q=d+2q+1\le2(d+q)$,
\begin{equation}
  m<\frac{6\Lambda^2(d+q)}{\varepsilon^2}<\frac{306456\,(d+q)}{\varepsilon^2}.
  \label{eq:m-rounding}
\end{equation}
This proves \eqref{eq:explicit-size} and completes the proof of
Theorem~\ref{thm:main}.

\section{Scope and limitations}
\label{sec:scope}

The proof uses, in order: the coupling and moment transfer of
Section~\ref{sec:convex-transfer}; the vacuum-moment identity
\eqref{eq:vacuum-moment} and the band inventory \eqref{eq:exact-inventory};
the light-sector inequality, Lemma~\ref{lem:LSL}; the row-local estimates
of Section~\ref{sec:hop}; the block Cauchy--Schwarz bounds of
Section~\ref{sec:band-bounds}, summarized in
Proposition~\ref{prop:band-envelope}; and the moment expansion, Markov
inequality, and rounding of Section~\ref{sec:proof}.  Only the last step
uses the specific parameter choice \eqref{eq:integer-parameters}, and only
through \eqref{eq:scale-conditions} and \eqref{eq:b-at-least-two}.

The following restrictions are part of the statement.
\begin{itemize}
\item All hash values and signs are independent.  The $2q$-th moment of the
  Gram error involves up to $4q$ column selectors, so moment matching of
  order $2q$ at the level of the Gram error does not follow from $2q$-wise
  independence of the primitive hashes and signs; a limited-independence
  version would require a separate joint-moment computation for the chosen
  encoding.
\item The matrix $U$ is fixed, or independent of $\Pi$
  (Remark~\ref{rem:independent-U}); an adaptively chosen $U$ is not
  covered.
\item The result concerns the SparseStack distribution
  \eqref{eq:sparsestack}, not an arbitrary sparse embedding with $s$
  nonzeros per column.
\item The constants have not been optimized.  The dominant losses are the
  factor $c_\star^2<2.51$ from Proposition~\ref{prop:transfer}, the factor
  $3$ from the number of bands, and the factor $3$ from the three terms of
  the envelope; together with $\Cband$ they account for the choice
  $\Lambda=90c_\star^2$ in \eqref{eq:global-constants}.
\end{itemize}

\appendix

\section{Verification status}
\label{app:lean}

\paragraph{Version history.}
Versions through Version~1.6.1 circulated under the title ``A Sparse-Fock
Proof of the Upper Edge for Fully Independent SparseStack.''  The main
theorem, parameter choices, final band envelope, and final numerical constants
are unchanged; Version~1.8 reorganizes and expands the exposition.
Version~2.0 revises the abstract and adds related-work attributions in
Sections~1.3, \ref{sec:convex-transfer}, and
\ref{sec:light-sector}; the mathematical results and proofs are unchanged.

The fixed-matrix theorem of this paper has been formalized in Lean~4
\cite{deMouraUllrich2021}.  The artifact uses Lean~4.33.0 and mathlib
\cite{Mathlib2020} commit
\texttt{db584cd6d46c92f209a44c0f1c829460d327499d}.  The statements below
concern that artifact: ``kernel-checked'' means that Lean accepted the stated
declarations after elaboration, not that Lean checked the prose of this paper.

\paragraph{Availability.}
The \href{https://github.com/DiarHaidary/nelson-nguyen-sparse-fock}
{project repository} contains the canonical Version~1.5 manuscript, the formal
source, and the audit record.  This correspondence audit used immutable
repository commit
\href{https://github.com/DiarHaidary/nelson-nguyen-sparse-fock/tree/b9da2b5d96e5fd70252e9c6551aefc92643e7d0b}
{\texttt{b9da2b5d96e5}} (the link records the full hash).  From the repository
root, the artifact builds with \texttt{cd formal \&\& lake build}.

\paragraph{Main statement.}
The final matrix-level declaration is
\nolinkurl{SparseFock.MainTheorem.fully_independent_sparseStack_matrix}.
Its hypotheses are $d\ge1$, $0<\varepsilon\le1$, $0<\delta\le1$, a real
matrix $U$, and $U^\T U=I_d$; its conclusion packages the SparseStack law,
the failure bound \eqref{eq:main-prob}, the two-sided event \eqref{eq:ose},
exact column sparsity, the bounds \eqref{eq:explicit-size}, $b\ge2$, and
$m=sb$.  The declaration has no further hypotheses.  The source signature is
given in
\href{https://github.com/DiarHaidary/nelson-nguyen-sparse-fock/blob/b9da2b5d96e5fd70252e9c6551aefc92643e7d0b/formal/SparseFockFormal/MainTheorem.lean}
{\texttt{formal/SparseFockFormal/MainTheorem.lean}}.
The frame-level endpoint is
\nolinkurl{SparseFock.MainTheorem.fully_independent_sparseStack}, and the
finite-law conditioning corollary for a random frame independent of the sketch
(Remark~\ref{rem:independent-U}) is
\begin{center}
\nolinkurl{SparseFock.MainTheorem.fully_independent_sparseStack_independent_frame}.
\end{center}

\paragraph{Axiom audit.}
Running \texttt{\#print axioms} on the main declaration reports only
\texttt{propext}, \texttt{Classical.choice}, and \texttt{Quot.sound}.  A
source sweep confirms that the project contains none of \texttt{sorry},
\texttt{admit}, project-defined \texttt{axiom}, \texttt{unsafe},
\texttt{partial}, \texttt{opaque}, \texttt{implemented\_by}, and
\texttt{native\_decide}.

\paragraph{Correspondence with the paper.}
The Lean artifact was developed against Version~1.5.  Direct comparison with
Theorem~\ref{thm:main} in Version~2.0 confirms the same real, fully independent
signed-hash SparseStack law and normalization; the same definitions of
$q,D_q,s,M_0,b,m$; the same constants $c_\star,\Cband,\Lambda$; and the same
operator-failure event, two-sided OSE event, exact column sparsity, strict
dimension bounds, $b\ge2$, and $m=sb$.  Version~2.0 states $n\ge1$
explicitly; in the Lean theorem this follows from $d\ge1$ and $U^\T U=I_d$.

The proof organization is not line-by-line identical.  The three directional
uses and two heavy-leg uses of the merged row-local hop
Lemma~\ref{lem:hop} are kernel-checked separately in
\texttt{DirectionalConcrete}, \texttt{DirectionalNormalized}, and
\texttt{HeavyBandsConcrete}.  The row-operator notation
$\Scal_r(Z),\Gcal(Z)$ of Section~\ref{sec:row-operators} packages concrete
factorizations checked in \texttt{TypedBlockCS} and the concrete heavy, light,
and mixed band modules.  The general vacuum-moment identity is
formalized; Example~\ref{ex:k2} is explanatory.  The final band envelope is
checked in \texttt{ConcreteBandEnvelope}, and the moment expansion, Markov
step, Gram/OSE equivalence, and integer rounding are checked in the endgame
modules.  Thus the fixed-matrix theorem and every load-bearing concrete
estimate used to prove it are kernel-checked, while some generic notation,
merged lemmas, examples, and sharper non-load-bearing remarks are expository
consolidations in this version.

\paragraph{What is not verified.}
The kernel check certifies the formal derivation.  The theorem-level
correspondence reported above is a human comparison, not a judgment made by
Lean.  Lean does not certify the prose, the explanatory examples, attribution,
novelty, or the claim that the formal statement captures the intended external
problem; those remain matters for conventional review.

\section{Notation index}
\label{app:notation}

\begin{center}
\small
\renewcommand{\arraystretch}{1.15}
\begin{tabular}{>{\raggedright\arraybackslash}p{0.24\textwidth}
                >{\raggedright\arraybackslash}p{0.68\textwidth}}
\toprule
Symbol & Meaning\\
\midrule
\multicolumn{2}{l}{\emph{Model and parameters}}\\
$n,d,U,u_i,U_{ij}$ & Ambient dimension, subspace dimension, fixed
$U\in\R^{n\times d}$ with $U^\T U=I_d$, its rows $u_i\in\R^d$, and
$U_{ij}=u_iu_j^\T$.\\
$\varepsilon,\delta,q,D_q,M_0$ & Distortion, failure probability,
$q=\max\{1,\lceil\log_2(d/\delta)\rceil\}$, $D_q=d+2q+1$,
$M_0=\Lambda^2D_q/\varepsilon^2$.\\
$s,b,m,h$ & Blocks (nonzeros per column), rows per block, total rows
$m=sb$, and $h=b-1$.\\
$H_g,h_g(i),\sigma_{gi},\Pi$ & CountSketch blocks, hash rows, signs, and the
SparseStack matrix \eqref{eq:sparsestack}.\\
$c_\star,\Cband,\Lambda$ & The constants \eqref{eq:global-constants}.\\
$\xi_{gi},\zeta_{gi},p_0,c_b$ & Signed one-hot selector, its
independent-entry proxy, $p_0=\Prob\{\zeta=0\}$, and $c_b=(1-p_0)^{-1}$.\\
$\eta_{ri}$ & Three-point entries \eqref{eq:eta-law} indexed by sites.\\
$E_\SStack,E_\iid$ & Hollow Gram errors \eqref{eq:ss-hollow} and
\eqref{eq:iid-row}.\\
\midrule
\multicolumn{2}{l}{\emph{Tensor-product model}}\\
$\sfI,(r,i)$ & Site set $[m]\times[n]$ and a site.\\
$\Fcal,\Omega,e_0,e_1,e_2$ & Product space $(\R^3)^{\otimes mn}$, vacuum,
and the local basis \eqref{eq:one-site-basis}.\\
$J,P,P^\dagger,R,R^\dagger$ & Jacobi matrix \eqref{eq:Jacobi} and its four
legs \eqref{eq:local-legs}; $Z_{ri}$ is the lift of $Z$ to site $(r,i)$.\\
$\tau,|\tau\rangle,\gr(\tau)$ & Pattern, basis vector, and grade
\eqref{eq:pattern-grade}.\\
$S_j(\tau),S_j^r(\tau),F_r(\tau)$ & Level-$j$ sites, level-$j$ columns of
row $r$, free columns of row $r$.\\
$\tau[a\mapsto j]$ & Pattern with site $(r,a)$ reassigned to level $j$
(row $r$ fixed by context).\\
$\nu_r(\tau),\Ncal_r$ & Grade of row $r$ and the corresponding operator.\\
$\Hcal_\nu,\Kcal_\nu,\pi_\nu,\Pcal_\nu$ & Grade-$\nu$ sector of $\Fcal$,
$\Kcal_\nu=\R^d\otimes\Hcal_\nu$, and the orthogonal projections onto
them.\\
$\Mcal,B^+,B^0,B^-$ & Operator \eqref{eq:M} and its bands \eqref{eq:bands}.\\
$\sfK_{T,\ell},|S;T\rangle,F_T$ & Block with heavy set $T$ and $\ell$
light sites, its basis vectors, and $\sfI\setminus T$.\\
\midrule
\multicolumn{2}{l}{\emph{Operators}}\\
$L_\bullet,H_\bullet,X_\bullet,Y_\bullet$ & Families
\eqref{eq:light-heavy-defs}--\eqref{eq:mixed-defs} in the inventory
\eqref{eq:exact-inventory}.\\
$\Scal_r(Z),\Gcal(Z)$ & Row operator and Gram operator
\eqref{eq:row-operator}--\eqref{eq:gram-operator}.\\
$\Ghat,H'$ & $\Gcal(P^\dagger)$ and the transpose light hop \eqref{eq:Ghat}.\\
$D_0,D_1,D_H$ & Diagonal operators \eqref{eq:D-ops}.\\
$\rho_\nu,\gamma_\nu,\beta_q$ & Envelope scales \eqref{eq:envelope-scales}
and the ladder constant \eqref{eq:beta}.\\
\bottomrule
\end{tabular}
\end{center}

\phantomsection

\end{document}